\documentclass[runningheads,envcountsect,envcountsame,orivec]{llncs}
\usepackage[a4paper,margin=25mm,includehead,headheight=13pt,headsep=5mm]{geometry}
\usepackage[T1]{fontenc}
\usepackage{amsmath,amssymb}
\numberwithin{equation}{section}
\usepackage{graphicx,pgfplots,tikz}
\pgfplotsset{compat=1.18} 
\usepackage{hyperref}
\usepackage{booktabs,tabularx}

\newcommand{\R}{\mathbb R}
\newcommand{\Z}{\mathbb Z}
\newcommand{\E}{\mathbb E}
\newcommand{\PP}{\mathbb P}

\newcommand{\dd}{\,\mathrm d}
\newcommand{\ind}{\mathbf 1}
\newcommand{\Gr}{\operatorname{Gr}}
\newcommand{\rank}{\operatorname{rank}}
\newcommand{\norm}[1]{\lVert#1\rVert}
\spnewtheorem{heuristicassumption}{Assumption}{\bfseries}{\itshape}

\title{Asymptotics of Random-Lattice Enumeration: \\Beyond HKZ via Dense Sublattices}
\titlerunning{Asymptotics of Random-Lattice Enumeration}
\author{Thijs Laarhoven}
\authorrunning{T. Laarhoven}
\institute{}

\begin{document}
\maketitle

\begin{abstract}
To bridge the gap between worst-case and heuristic analyses of lattice enumeration, we study the sizes of (unpruned) enumeration trees on $n$-dimensional Haar-random lattices. Writing $n^{c n + o(n)}$ for the tree size at the shortest vector radius, we investigate the asymptotic scaling of $c$, and which bases minimize it.

For quasi-HKZ bases, we rigorously analyze the random-lattice HKZ geometry and establish $c = 1/(2e) \approx 0.1839$. This matches the worst-case enumeration exponent of Hanrot--Stehl\'{e} and confirms earlier predictions for the geometry of HKZ bases on random lattices.

Searching beyond quasi-HKZ bases, we prove that with high probability, the minimum achievable tree size exponent $c$ over all bases of a random lattice satisfies:
\begin{align}
0.1464 \approx \frac{2 - \sqrt{2}}{4} \leq c < 0.1474.
\end{align}

For the lower bound, we use Kim's mean-value formula for nested chains of primitive sublattices to rule out bases with small prefix determinants. This strengthens the earlier Rankin-based heuristic lower bound exponent $c \geq 1/8$, discussed by Nguyen and Hanrot--Stehl\'{e}, and rules out the existence of random-lattice bases with trees as small as predicted by the geometric series assumption.

For the upper bound, we show that with high probability a dense primitive half-rank sublattice exists with determinant at most $n^{n/8 + o(n)}$. Combining dual HKZ on this sublattice with HKZ on its quotient gives a basis with $c \leq 1/(4\sqrt{e}) \approx 0.1516$. Recursively splitting each half lowers the upper bound below $0.1474$, and under a natural heuristic assumption we prove that the lower bound can be attained with suitable sublattice splits.

Combined with the collision search of Feng--Gao--Liu, our two-block construction gives a polynomial-space algorithm solving the closest vector problem with preprocessing (CVPP) on random lattices in query time $n^{0.0892 n + o(n)}$ for arbitrary targets.
\keywords{Random lattices \and Enumeration \and HKZ reduction \and Dense sublattices}
\end{abstract}


\section{Introduction}

Currently, one of the main challenges in cryptography is the migration to quantum-safe cryptography. After decades of research, lattice-based cryptography has emerged as a prime candidate due to its favorable performance trade-offs over other proposals. The security of lattice-based schemes relies on the hardness of high-dimensional lattice problems such as the shortest vector problem (SVP) and the closest vector problem (CVP). A good understanding of the best algorithms for solving these problems is critical for designing secure and efficient lattice-based cryptographic primitives.

For solving hard lattice problems, perhaps the most important tool in the cryptanalyst's toolbox is the BKZ algorithm: a blockwise basis reduction algorithm turning a "bad" lattice basis into a "good" one~\cite{schnorreuchner,chennguyen}. The BKZ algorithm relies on an exact SVP subroutine to do the heavy lifting; within lower-dimensional projected sublattices (blocks), an SVP oracle finds a shortest lattice vector, which is then used to improve the overall basis. The block size controls the cost of the algorithm and the quality of the basis, and the overall costs of BKZ are often estimated to be proportional to solving SVP in random lattices whose dimension equals the block size. As such, the costs of exact SVP oracles directly impact the estimated costs of BKZ.

For solving exact SVP, there are two main lines of work. Lattice enumeration~\cite{finckepohst,schnorreuchner,gnr,chennguyen} can be seen as a brute-force search through all potential combinations of basis vectors that might produce a shortest vector. By enumerating all combinations and keeping track only of the shortest vector found so far, this algorithm only requires polynomial memory, but its running time in dimension $n$ scales as $n^{\Theta(n)}$~\cite{kannan83,kannan87}. Its competitor, lattice sieving~\cite{ajtai01,g6k,bdgl}, is a saturation algorithm which first samples a lot of lattice vectors, and then repeatedly tries to form shorter combinations until it finds a shortest vector. This method runs in time $2^{\Theta(n)}$, but additionally requires $2^{\Theta(n)}$ space for storing all the sampled vectors. Currently, the highest records in the SVP challenge database~\cite{svpchallenge} were achieved using lattice sieving. However, for a long time lattice enumeration clearly dominated this same database, and improvements to enumeration may shift the balance back in its favor.

In this work, we study the main ingredient dominating the costs of lattice enumeration: the size of the full (unpruned) enumeration tree when searching at a radius proportional to the shortest vector norm. In particular, we focus on the sizes of these trees on Haar-random lattices of dimension $n$ and determinant one. For a basis $\mathbf{B}$ of a random lattice $\mathcal{L}$, write:
\begin{align}
T(\mathbf{B}) = \#\{\text{nodes in the full enumeration tree at radius }\lambda_1(\mathcal{L})\} = n^{c n + o(n)}.
\end{align}
We are interested in finding bases $\mathbf{B}$ which minimize $c$ for Haar-random lattices.\footnote{Note that for random lattices, with high probability $\lambda_1(\mathcal{L}) = (1 + o(1)) \cdot \sqrt{n / (2 \pi e)}$. As the leading constant $c$ in the enumeration tree size does not change when multiplying the radius with $1 + o(1)$, we can equivalently replace $\lambda_1(\mathcal{L})$ in the definition of $T(\mathbf{B})$ and $c$ with the Gaussian heuristic.}


\subsection{Related work}

A classical approach for obtaining a basis with small enumeration trees is HKZ reduction~\cite{kz1873}. To reduce an $n$-dimensional basis, this greedy approach selects a shortest non-zero vector as its first basis vector, projects away its span, and calls HKZ on the remainder (and lifts the result back to the original lattice). Hanrot--Stehl\'{e} established that, even for worst-case lattices, the tree size exponent for (quasi-)HKZ bases satisfies $c \leq 1/(2e)$~\cite{hs2007,hscomplete}. The same authors showed that this worst-case bound is tight~\cite{hsworst} and conjectured that the same exponent applies to random lattices~\cite[Section~6]{hscomplete}. The heuristic HKZ model of Hanrot--Pujol--Stehl\'{e}~\cite[Section~4.1]{hps} and its subsequent refinement by Haque--Pieprzyk~\cite{haque} support this prediction, but rely on unproven heuristic assumptions about the geometry of HKZ-reduced bases. Aono--Nguyen--Seito--Shikata~\cite{anss} experimentally observed that HKZ bases may be suboptimal for minimizing the tree sizes. In general, it is not known whether HKZ bases are asymptotically optimal (up to lower order terms), or if bases exist achieving $c < 1/(2e) - \varepsilon$ for fixed $\varepsilon > 0$ and large $n$.

When considering enumeration trees for other bases, past results have hinted that one may hope to do better. The best known heuristic lower bound of Nguyen for arbitrary bases~\cite[``A Heuristic Lower Bound'']{nguyenhermite} says $c \geq 1/8$ for all bases, which leaves a clear gap with the heuristic estimate of $c = 1/(2e)$ for HKZ bases. In the idealized setting of a basis following the geometric series assumption~\cite{schnorrgsa}, one would actually obtain $c = 1/8$~\cite[Section~6, p.~27]{hscomplete}. This reduced tree size for GSA bases was cleverly used by Albrecht et al.~\cite{faster} to obtain a faster BKZ algorithm. However, for random lattices we generally do not expect bases with such a GSA-profile (with $c = 1/8$) to exist.

A concrete direction for finding bases with smaller enumeration trees is based on finding dense sublattices. This approach was already suggested by Gama--Howgrave-Graham--Koy--Nguyen and Hanrot--Stehl\'{e}~\cite{ghkn2006,hscomplete}: search for a half-rank sublattice with small determinant, split the lattice basis in two halves, and reduce each part separately. The same relation between small enumeration trees and dense sublattices was observed by Li--Nguyen~\cite{linguyen14}. An open question is whether using dense sublattices can help find bases with strictly smaller $c$ than HKZ. 

For using dense sublattices in cryptanalytic algorithms, Ducas--Loyer~\cite[Sections~3--6]{ducasloyer} recently studied BKZ with low-rank densest-sublattice oracles, especially ranks $2$ and $3$, using Kim's sublattice statistics~\cite{kimflags} to predict the output slope. Their heuristic estimates give only small gains, insufficient to compensate for the loss of dimensions for free, and their experiments show that these estimates are optimistic. Their experiments further reveal strong clustering in dense-sublattice counts, illustrating the difficulty of inferring typical existence from first moments.


\subsection{Contributions}

In this work, we answer several of the open questions above. Can we rigorously prove asymptotics on $c$ for quasi-HKZ bases on random lattices? Are HKZ bases asymptotically optimal, or can we find better bases with lower asymptotics for $c$? And if so, how small can $c$ be, and can we prove sharper lower bounds on $c$ than the heuristic bound $c \geq 1/8$?

\paragraph{Random-lattice HKZ geometry.}

Our first main contribution confirms heuristic predictions about the HKZ geometry of Hanrot--Pujol--Stehl\'{e} and Haque--Pieprzyk~\cite{hps,haque}, as well as the conjecture of Hanrot--Stehl\'{e}~\cite{hscomplete}: for random lattices, quasi-HKZ bases have trees of size $n^{n/(2e)+o(n)}$. More precisely: with high probability over the randomness of $\mathcal{L} \sim \mu_n$, with $\mu_n$ denoting the Haar probability measure over the space $\mathrm{SL}_n(\R)/\mathrm{SL}_n(\Z)$ of $n$-dimensional lattices of determinant one, \textit{all} quasi-HKZ bases of $\mathcal{L}$ have this tree size scaling. To prove this, we show that the Gram--Schmidt profile of quasi-HKZ bases on random lattices essentially matches what heuristics predicted, from which the constant $c$ follows. This matches the asymptotics of Hanrot--Stehl\'{e} for the sizes of enumeration trees for quasi-HKZ bases on worst-case lattices.

\begin{theorem}[Tree sizes of all quasi-HKZ bases]
\label{thm:enumeration}
With probability $1 - o(1)$ over $\mathcal{L} \sim \mu_n$, every quasi-HKZ basis $\mathbf{B}_{\mathrm{HKZ}}$ of $\mathcal{L}$ has enumeration tree size:
\begin{align}\label{eq:main-enum}
T(\mathbf{B}_{\mathrm{HKZ}}) = n^{c_{\mathrm{HKZ}} \cdot n + o(n)}, \qquad c_{\mathrm{HKZ}} = \frac{1}{2e} = 0.18393\dots.
\end{align}
\end{theorem}

\paragraph{Beyond HKZ via dense sublattices.}

Whereas Theorem~\ref{thm:enumeration} proves what heuristics already suggested, to the best of our knowledge the following concrete asymptotic improvement has not been established before: (1) showing that useful dense sublattices exist with high probability, and (2) showing that with these dense sublattices, one can indeed go strictly below the $c_{\mathrm{HKZ}} = 1/(2e)$ exponent of quasi-HKZ bases, thereby showing quasi-HKZ bases are asymptotically suboptimal in the leading constant $c$. This idea was previously suggested by Hanrot--Stehl\'{e}~\cite{hscomplete}, and here we rigorously establish the asymptotic improvement over quasi-HKZ bases. Using a single dense primitive sublattice of half the rank, and optimizing for the volume of this sublattice so that with high probability it exists, we obtain an explicit limiting Gram--Schmidt profile and asymptotic tree size constant $c_2$. 

\begin{theorem}[Tree size of a two-block split]
\label{thm:two-way}
With probability $1 - o(1)$, $\mathcal{L} \sim \mu_n$ has a primitive sublattice $\mathcal{M}$ of rank $m = \lfloor(n-1) / 2\rfloor$ with $\det\mathcal{M} \leq e^{O(n)} n^{n/8}$. Combining a reversed dual-HKZ basis of $\mathcal{M}$ with lifts of an HKZ basis of its orthogonal quotient gives a basis $\mathbf{B}_2$ of $\mathcal{L}$ with enumeration tree size:
\begin{align}\label{eq:two-way-main}
T(\mathbf{B}_2) = n^{c_2 \cdot n + o(n)}, \qquad c_2 = \frac{1}{4\sqrt e} = 0.15163\dots.
\end{align}
\end{theorem}

Note that if we split a typical HKZ basis halfway, it has a sublattice $\mathcal{M}$ with determinant $\det \mathcal{M} = n^{((\ln 2) / 4 + o(1)) n}$. This is exactly the boundary where a denser sublattice would lead to smaller trees. As finding dense sublattices is difficult, a first goal for a faster SVP algorithm may be to find \textit{any} half-rank sublattice with density higher than what we get from HKZ.

A natural question to ask then is: why stop at a two-way split? Can we recursively look for dense sublattices in each half, and further reduce the constant $c$? Or can we immediately go for multi-block splits, to obtain smaller trees? Using recursive binary splits, the constant can be lowered below $0.1474$, with the limit numerically approaching $0.14739700\dots$.

\begin{theorem}[Tree size of recursive binary splits]
\label{thm:eight-blocks}
With probability $1 - o(1)$, $\mathcal{L} \sim \mu_n$ contains a sequence of dense sublattices $\mathcal{M}_0 \subseteq \mathcal{M}_1 \subseteq \dots \subseteq \mathcal{L}$ such that, combining (lifted) HKZ and reversed dual-HKZ bases of (orthogonal quotients of) these blocks, we obtain a basis $\mathbf{B}_{\mathrm{rec}}$ of $\mathcal{L}$ with enumeration tree size:
\begin{align}\label{eq:eight-block-main}
T(\mathbf{B}_{\mathrm{rec}}) \leq n^{c_{\mathrm{rec}} \cdot n + o(n)}, \qquad c_{\mathrm{rec}} < 0.14740.
\end{align}
\end{theorem}

\paragraph{Heuristic sublattice splits.}

For the previous upper bounds, it is important to note that these are based on ``safe'' splits: for every binary split, we make sure that with high probability the desired dense sublattice exists. Better subdivisions may be possible if we allow splits which are exponentially rare, but where there are exponentially many potential sublattices to consider. However, analyzing subdivisions with rare sublattices rigorously appears challenging, as it requires controlling correlations between different sublattices. For instance, showing that with the subdivision strategy there are $n^{0.002 n^2}$ potential sublattices, each having a chance of $n^{-0.001 n^2}$ of being sufficiently dense, is not sufficient to show that one dense sublattice exists with high probability, since sublattices are not independent.

In general, it may be natural to assume that if the first moments of the probabilities align when averaged over all random lattices, then we expect such towers of primitive sublattices to exist for most random lattices. Below we informally state this heuristic assumption; a precise definition follows later.

\begin{heuristicassumption}[First-moment sufficiency, informal]
\label{heur:assumption}
Consider nested chains of primitive sublattices, with prescribed ranks and determinant bounds. Suppose that, for some $\delta > 0$, the expected number of such chains in a Haar random lattice is at least $n^{\delta n^2 + o(n^2)}$, and the same bound holds for every non-empty subchain obtained by omitting intermediate sublattices. Then a Haar-random lattice contains such a complete chain with probability $1 - o(1)$.
\end{heuristicassumption}

Ducas--Loyer \cite[Section 6.2]{ducasloyer} observed strong clustering in rank-two dense-sublattice counts at determinant thresholds where the expected count is one. Assumption~\ref{heur:assumption} concerns a different asymptotic regime: a fixed number of prescribed sublattice ranks, each proportional to $n$, with the expected count of every nonempty subchain at least $n^{\delta n^2}$ for some fixed $\delta > 0$ and all sufficiently large $n$. The assumption only predicts that the complete chain exists with high probability, without assuming Poisson statistics or concentration of the counts around their means. The observations of~\cite{ducasloyer} do not contradict this statement, but emphasize that proving the existence of such chains in general may require justification beyond first moments.

Under this assumption, we show that we can slightly lower the previous upper bounds. Proving this result without the need for the heuristic assumption is left as an open problem.

\begin{theorem}[Tree size of improved heuristic splits]
\label{thm:heuristic}
Suppose Assumption~\ref{heur:assumption} holds, and let $\varepsilon > 0$ be fixed. Then with probability $1 - o(1)$, $\mathcal{L} \sim \mu_n$ contains a basis $\mathbf{B}$ satisfying:
\begin{align}\label{eq:heuristic}
T(\mathbf{B}) \leq n^{(c_{\mathrm{flat}} + \varepsilon) \cdot n + o(n)}, \qquad c_{\mathrm{flat}} = \frac{2-\sqrt2}{4} = 0.14644\dots.
\end{align}
\end{theorem}

\paragraph{A sharp(er) tree size lower bound.}

To understand whether our improved upper bounds are (close to) optimal, we derived an improved asymptotic lower bound on $c$ which with high probability over the randomness of the lattice holds simultaneously for all bases of this lattice. The previous best (heuristic) lower bound described by Nguyen~\cite{nguyenhermite} was of the form $c \geq 1/8 = 0.125$; we show this bound can be increased to $c \geq (2 - \sqrt{2})/4 \approx 0.146$. This leaves a gap of slightly less than $0.001$ in the constant $c$ between our best proven upper and lower bounds, and this matches our best heuristic upper bound.

\begin{theorem}[Lower bound on the tree size]
\label{thm:any-basis}%
With probability $1 - o(1)$ over $\mathcal{L} \sim \mu_n$, every basis $\mathbf{B}$ of $\mathcal{L}$ satisfies:
\begin{align}\label{eq:any-basis}
T(\mathbf{B}) \geq n^{c_{\mathrm{flat}} \cdot n + o(n)}, \qquad c_{\mathrm{flat}} = \frac{2-\sqrt2}{4} = 0.14644\dots.
\end{align}
\end{theorem}

If the heuristic assumption holds, then this lower bound is asymptotically tight. Without this assumption, our upper and lower bounds leave a gap of less than $0.001$ on the optimal $c$. 

\paragraph{Collision search beyond HKZ.}

In recent work, Feng--Gao--Liu~\cite{fgl2026} showed how to replace exhaustive enumeration over the entire enumeration tree by a polynomial-space collision search, to find a pair of vectors whose difference is a shortest vector. The relevant geometric factor in their analysis is $\sqrt{\mathcal P(\mathbf B)}$, where $\mathcal P(\mathbf B) = \prod_i \max\{1, \|\mathbf b_i^*\|^{-1}\}$ for determinant-one lattices. Using quasi-HKZ bases, they then achieve an SVP cost of:
\begin{align}
\sqrt{\mathcal P(\mathbf B)} \leq n^{\tau_{\mathrm{HKZ}} \cdot n + o(n)}, \qquad \tau_{\mathrm{HKZ}} = \frac{1}{4e} = 0.09196\dots.
\end{align}
We show that our lower bound carries over to their framework, with an additional factor $1/2$. This bounds the geometric factor in the FGL analysis; a lower bound on the algorithm's actual running time would require a separate argument. 

\begin{theorem}[Lower bound on the Feng--Gao--Liu cost factor]
\label{thm:fgl-lower-note}
With probability $1 - o(1)$ over $\mathcal L \sim \mu_n$, simultaneously for every basis $\mathbf B$ of $\mathcal L$,
\begin{equation}
\sqrt{\mathcal P(\mathbf B)} \geq n^{\tau_{\mathrm{low}} \cdot n + o(n)}, \qquad \tau_{\mathrm{low}} = \frac{2 - \sqrt2}{8} = 0.07322\dots. \label{eq:fgl-lower-note}
\end{equation}
\end{theorem}

To find the optimal basis minimizing the dominant cost, it is important to note that its asymptotic scale is controlled by the total variation of the normalized prefix profile, rather than its maximum height. Thus minimizing the largest enumeration level in the tree and minimizing the costs of this collision-search approach are different optimization problems. With the same binary split procedure, we do immediately obtain an improved cost estimate.

\begin{theorem}[FGL costs of a two-block split]
\label{thm:fgl-two-way}
With probability $1 - o(1)$, $\mathcal{L} \sim \mu_n$ contains a basis $\mathbf{B}_2$ with FGL costs:
\begin{align}
\sqrt{\mathcal P(\mathbf B_2)} \leq n^{\tau_2 \cdot n + o(n)}, \qquad \tau_2 = \frac{1}{4 \sqrt{e}} - \frac{1}{16} = 0.08913\dots. \label{eq:fgl-two-way}
\end{align}
\end{theorem}

Note that this does not immediately give a faster polynomial-space SVP algorithm than Feng--Gao--Liu; that would require an efficient algorithm for finding the half-rank dense sublattice that is needed to get the better basis profile.

For the FGL cost function, we observed numerically that repeated splits beyond one single symmetric split do not decrease the constant $\tau$. Extending our search to splits satisfying Assumption~\ref{heur:assumption}, we did find any better profiles than the two-block profile from Theorem~\ref{thm:fgl-two-way}. We conjecture that the two-block split may actually be optimal for minimizing the dominant collision search costs, and leave a proof (or a counterexample with a better basis profile) for future work.

\paragraph{Faster polynomial-space CVPP.}

Since finding dense sublattices appears to be a very hard problem, it is unclear whether the above results can be used for faster enumeration algorithms for SVP. Dadush--Micciancio~\cite[Theorem~3.1]{dadushmicciancio} gave an exact algorithm for the densest rank-$k$ sublattice problem, which for $k = \Theta(n)$ runs in time $n^{O(n^2)}$ and space $2^{O(n)}$. One would first have to reduce the costs below $n^{n/(4e) + o(n)}$ time and $n^{O(1)}$ space to make this approach suitable for solving SVP faster than Feng--Gao--Liu. 

As a positive result, we can show that with the above results we obtain a faster polynomial-space algorithm for the closest vector problem \textit{with preprocessing} (CVPP): if we do not count the preprocessing costs, then we can use (an extension of) the algorithm of Dadush--Micciancio to find all the dense sublattices we need. For sufficiently large batches of closest vector problem instances on the same lattice, this gives a lower amortized time complexity per CVP target.

\begin{theorem}[Polynomial-space CVPP by collision search]
\label{thm:cvpp-collision}
With probability $1 - o(1)$ over $\mathcal L \sim \mu_n$, there is a basis $\mathbf B_2$ of $\mathcal{L}$ such that, for every target $\boldsymbol t \in \mathbb R^n$, a randomized algorithm finds an exact closest lattice vector with probability at least $2/3$, using polynomial space and query time:
\begin{equation}
\mathrm Q \leq n^{\tau_2 n + o(n)} \approx n^{0.08913n + o(n)}. \label{eq:cvpp-query-collision}
\end{equation}
\end{theorem}

The main results regarding the tree sizes and algorithm complexities are summarized and compared against related work in Tables~\ref{tab:enumeration-upper} and \ref{tab:cvpp-complexities}. 

Figure~\ref{fig:tree-profiles} further summarizes what enumeration trees look like for different bases, by plotting the logarithmic enumeration tree width against the prefix dimension (the reverse of the depth in the enumeration tree). The maximum of the curve determines the tree size exponent $c$, while the area under the curve is dictated by the area rule in Section~\ref{sec:profile-area}, which essentially says the volume under the curve is at least $1/8$. 

The lower bound has a geometric interpretation in terms of the normalized prefix determinants, which describe the logarithmic widths of the enumeration tree. Kim’s formula for nested sublattices~\cite{kimflags} implies an area rule: a fixed continuous upper bound on these profiles cannot have area less than $1/8$ on a typical Haar-random lattice. Suppose that a basis has leading tree-size coefficient at most $c < 1/4$. Retain its chain of prefixes between dimensions $2cn$ and $(1 - 2c) n$, and reduce the two outer blocks $[0, 2cn]$ and $[(1 - 2c)n, n]$ by reversed dual-HKZ and HKZ, respectively. This preserves the middle prefix determinants and makes sure the leading tree-size coefficient remains at most $c$. The deterministic HKZ bounds place the modified profile below a curve with a plateau of height $c$ and two curved flanks, taking (reversed) HKZ tail shapes. Relative to the rectangle of height $c$, the flanks remove a total area of $c^2$, leaving area $c-c^2$. The area rule therefore requires
\begin{align}
c - c^2 \geq \frac{1}{8}.    
\end{align}
Solving for $c$ yields the lower bound $c \geq (2 - \sqrt2)/4$. The flank bounds apply after this change of basis; they are not pointwise constraints on the original profile. At equality, the resulting curve is the plateau candidate shown in Figure 1.

Note that in Figure~\ref{fig:tree-profiles}, recursive binary splits typically fix the "dip" at $(1/2, 1/8)$, and consecutive splits do not alter this dip. Even with a large recursion depth, the gap between the binary splits and the plateau at $x = 1/2$ remains significant, explaining the gap between the lower bound and the best upper bound with recursive splits. The heuristic splits however avoid this dip halfway, and make sure we approximate the plateau well at every value $x$.

An alternative visual representation of the different profiles is given in Figure~\ref{fig:gs-profiles}, showing the Gram--Schmidt profiles at various (reversed) depths in the tree. This figure essentially shows the derivatives of the curves in Figure~\ref{fig:tree-profiles}. The two-block profile jumps from $-1/4$ to $1/4$ at the split, corresponding to a factor $n^{1/2+o(1)}$ between adjacent Gram--Schmidt lengths. Note that good enumeration bases do not necessarily begin with a short vector, and better bases may even have superpolynomial first-vector lengths. The endpoint prediction and a rigorous first-vector lower bound are discussed in Section~\ref{sec:split-gs}.

\begin{table}[!t]
\centering
\small
\setlength{\tabcolsep}{4pt}
\renewcommand{\arraystretch}{1.16}
\begin{tabularx}{\textwidth}{@{}>{\raggedright\arraybackslash}X
  >{\raggedright\arraybackslash}p{0.37\textwidth}r@{}}
\toprule
\textbf{Enumeration tree size bounds} & \textbf{Lattice model} & $\boldsymbol{c}$ \\
\midrule
HKZ bases~\cite{hps,haque}
  & Heuristic & $0.184$ \\
Plateau profile (\textbf{Theorem~\ref{thm:heuristic}})
  & Heuristic & $\mathbf{0.146}$ \\
Rankin-based bound~\cite{nguyenhermite}
  & Heuristic (all bases) & $\geq 0.125$ \\
\midrule
Quasi-HKZ bases (\textbf{Theorem~\ref{thm:enumeration}})
  & Haar-random lattices & $0.184$ \\
Two-block split bases (\textbf{Theorem~\ref{thm:two-way}})
  & Haar-random lattices & $0.152$ \\
Recursive binary split bases (\textbf{Theorem~\ref{thm:eight-blocks}})
  & Haar-random lattices & $\mathbf{0.147}$ \\
Arbitrary bases (\textbf{Theorem~\ref{thm:any-basis}})
  & Haar-random lattices (all bases) & $\geq \mathbf{0.146}$ \\
\midrule
Quasi-HKZ bases~\cite{helfrich}
  & Worst-case lattices & $0.500$ \\
Quasi-HKZ bases~\cite[Thm.~2.5]{hscomplete}
  & Worst-case lattices & $\mathbf{0.184}$ \\
HKZ worst-case examples~\cite[Thm.~3]{hsworst}
  & Worst-case lattices (HKZ bases) & $\geq 0.184$ \\
\bottomrule
\end{tabularx}
\vspace{0.05cm}
\caption{Bounds on the constant $c$, rounded to three decimals, in the enumeration tree size $T(\mathbf{B}) = n^{cn + o(n)}$. Lower bounds are highlighted as $\geq c$, and explicitly state for which bases this lower bound holds. These bounds concern the sizes of full, unpruned enumeration trees, at a radius proportional to the shortest vector length.}
\label{tab:enumeration-upper}
\end{table}

\begin{table}[!t]
\centering
\small
\setlength{\tabcolsep}{4pt}
\renewcommand{\arraystretch}{1.16}
\begin{tabularx}{\textwidth}{@{}>{\raggedright\arraybackslash}X
  >{\raggedright\arraybackslash}p{0.37\textwidth}r@{}}
\toprule
\textbf{Query times for poly-space CVPP} & \textbf{Lattice model} & $\boldsymbol{c}$ \\
\midrule
Approximate Voronoi cells~\cite{DLW19,cvpp-laarhoven21}
  & Heuristic & $0.500$ \\
Large-tuple sieving~\cite{cvpp-bls16}
  & Heuristic & $0.500$ \\
Enumeration: Dual-HKZ bases~\cite{hscomplete}
  & Heuristic & $0.184$ \\
Enumeration: Plateau profile (\textbf{Theorem~\ref{thm:heuristic}})
  & Heuristic & $0.146$ \\
Collision search: Dual-HKZ bases ~\cite{fgl2026}
  & Heuristic & $\mathbf{0.092}$ \\
\midrule
Dual BDD with perturbations~\cite{cvpp-drs14}
  & Haar-random lattices & $0.500$ \\
Enumeration: Recursive splits (\textbf{Theorem~\ref{thm:eight-blocks}})
  & Haar-random lattices & $0.147$ \\
Collision search: Two blocks (\textbf{Theorem~\ref{thm:cvpp-collision}})
  & Haar-random lattices & $\mathbf{0.089}$ \\
Collision search: Any basis (\textbf{Theorem~\ref{thm:fgl-lower-note}})
  & Haar-random lattices (FGL) & $\geq \mathbf{0.073}$ \\
\midrule
Enumeration: Dual-HKZ bases~\cite{cvpp-bloemer00,cvpp-wagnermaze14}
  & Worst-case lattices & $1.000$ \\
Compact representations~\cite{cvpp-hrs20}
  & Worst-case lattices & $1.000$ \\
Kannan enumeration~\cite{kannan87,hscomplete,cvpp-fu16}
  & Worst-case lattices & $\mathbf{0.500}$ \\
\bottomrule
\end{tabularx}
\vspace{0.05cm}
\caption{Bounds on the constant $c$, rounded to three decimals, in the query time $\mathrm{Q}(\mathbf{B}, \mathbf{t}) = n^{cn + o(n)}$ for CVPP algorithms using polynomial-size advice and query space. The lower bound is highlighted as $\geq c$, and concerns only the dominant FGL cost factor. The stated bounds include CVPP-adaptations of SVP and BDD algorithms, as well as heuristic extrapolations. Although our upper bounds hold for arbitrary targets, in all cases we assume the target $\mathbf{t}$ is uniform from $\R^n / \mathcal{L}$.}
\label{tab:cvpp-complexities}
\end{table}

\begin{figure}[!tbp]
\centering
\input{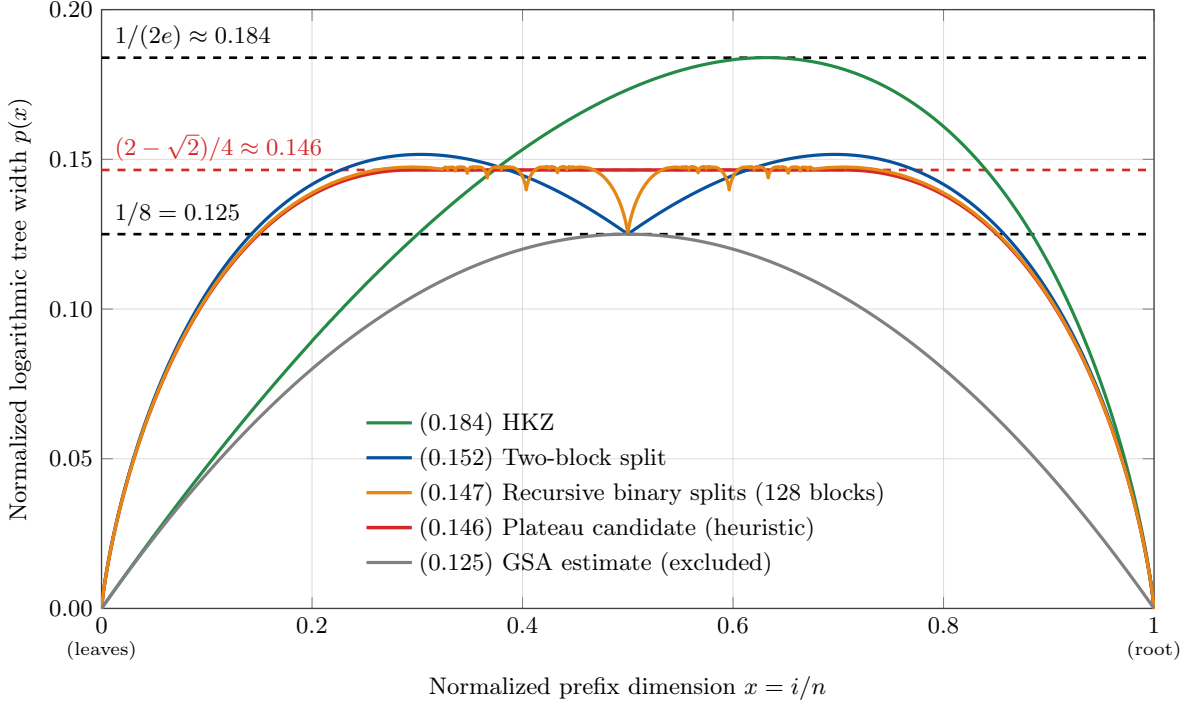}
\caption{Normalized tree width profiles. The solid red curve is the plateau candidate from \eqref{eq:plateau-profile}, whose attainability is only heuristically proved. All non-gray profiles have area $1/8$; the gray GSA estimate has area $1/12$ and is excluded by the area rule in Section~\ref{sec:profile-area}. The red curve is the derivative of the limiting prefix scaling; we do not prove that the actual Gram--Schmidt profiles converge to this curve.}
\label{fig:tree-profiles}
\end{figure}

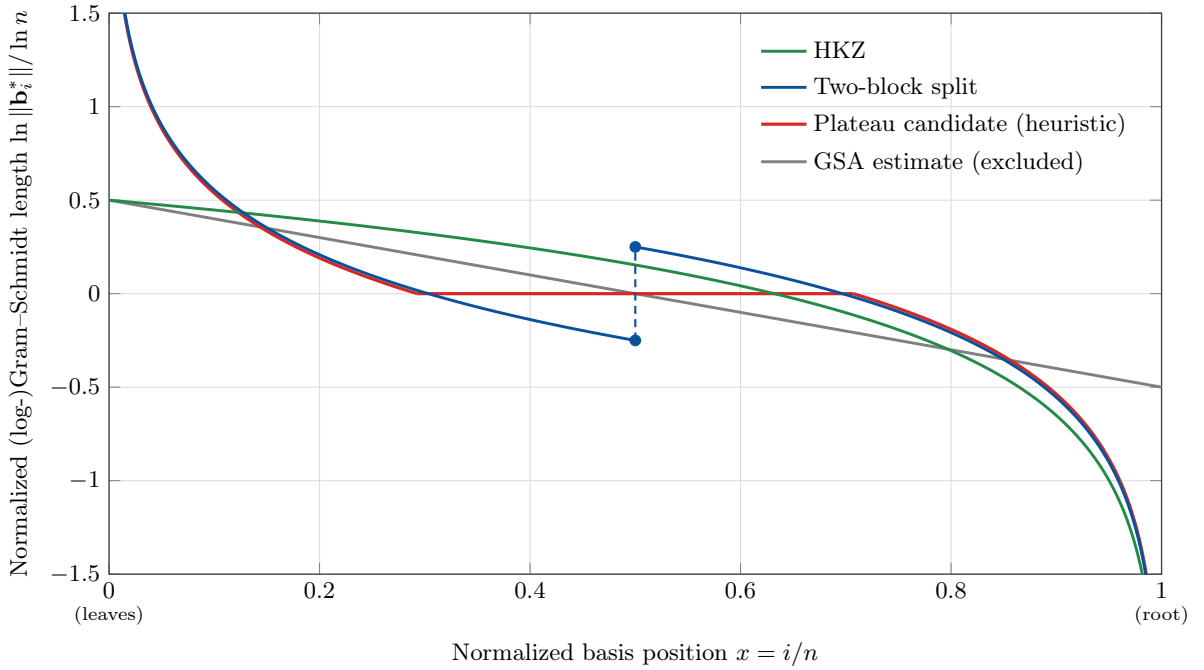
\begin{figure}[!tbp]
\centering
\begin{tikzpicture}
\definecolor{gshkz}{HTML}{238B45}
\definecolor{gstwo}{HTML}{08519C}
\definecolor{gsplateau}{HTML}{D62728}
\definecolor{gsgsa}{HTML}{808080}

\begin{axis}[
  width=15.5cm, height=9.0cm,
  axis lines=box, tick align=inside, tick pos=both,
  grid=major, major grid style={black!14,line width=0.3pt},
  axis line style={black!75,line width=0.5pt},
  xmin=0,xmax=1,ymin=-1.5,ymax=1.5,
  xtick={0,0.2,0.4,0.6,0.8,1},
  xticklabels={{\shortstack{$0$\\\scriptsize(leaves)}},$0.2$,$0.4$,$0.6$,$0.8$,{\shortstack{$1$\\\scriptsize(root)}}},
  ytick={-1.5,-1,-0.5,0,0.5,1,1.5},
  tick label style={font=\small},
  label style={font=\small},
  xlabel={Normalized basis position $x=i/n$},
  ylabel={Normalized (log-)Gram--Schmidt length $\ln \|\mathbf{b}_i^*\| / \ln n$},
  legend style={at={(0.98,0.97)},anchor=north east,
    draw=none,fill=none,font=\small,row sep=2pt},
  legend cell align=left,
  clip=true,
]
\addplot[gsgsa,line width=1.1pt,domain=0:1,samples=2,forget plot]
  {0.5-x};
\addplot[gsplateau,line width=1.2pt,domain=0.0003:0.29289321881345248,samples=500,forget plot]
  {0.5*ln(0.29289321881345248/x)};
\addplot[gsplateau,line width=1.2pt,forget plot]
  coordinates {(0.29289321881345248,0) (0.70710678118654752,0)};
\addplot[gsplateau,line width=1.2pt,domain=0.70710678118654752:0.9997,samples=500,forget plot]
  {-0.5*ln(0.29289321881345248/(1-x))};

\addplot[gshkz,line width=1.1pt,domain=0:0.9997,samples=500,forget plot]
  {0.5*(1+ln(1-x))};
\addplot[gstwo,line width=1.1pt,domain=0.0003:0.5,samples=500,forget plot]
  {-0.25-0.5*ln(2*x)};
\addplot[gstwo,line width=1.1pt,domain=0.5:0.9997,samples=500,forget plot]
  {0.25+0.5*ln(2*(1-x))};

\addplot[gstwo,densely dashed,line width=0.8pt,forget plot]
  coordinates {(0.5,-0.25) (0.5,0.25)};
\addplot[gstwo,only marks,mark=*,mark size=2pt,forget plot]
  coordinates {(0.5,-0.25) (0.5,0.25)};

\addlegendimage{gshkz,line width=1.1pt}
\addlegendentry{HKZ}
\addlegendimage{gstwo,line width=1.1pt}
\addlegendentry{Two-block split}
\addlegendimage{gsplateau,line width=1.2pt}
\addlegendentry{Plateau candidate (heuristic)}
\addlegendimage{gsgsa,line width=1.1pt}
\addlegendentry{GSA estimate (excluded)}
\end{axis}
\end{tikzpicture}
\caption{Normalized logarithmic Gram--Schmidt profiles. Logarithmic singularities are clipped by the vertical plotting window. The GSA estimate is impossible for random lattices.}
\label{fig:gs-profiles}
\end{figure}




\subsection{Techniques}

Our proofs combine uniform lattice-point estimates with averaging formulas for primitive vectors and sublattices.

\textit{From prefix determinants to enumeration widths.} Let $P_k$ denote the logarithmic determinant of the sublattice generated by the first $k$ basis vectors. We prove that, with high probability over the random lattice $\mathcal{L}$, and simultaneously for all bases $\mathbf{B}$ of $\mathcal{L}$, $\log T(\mathbf{B}) = \max_k P_k + O(n)$. The lower estimate follows by averaging lattice-point counts over translated balls; the upper estimate uses Gaussian mass and Poisson summation. This uniform comparison allows us to analyse bases chosen after observing the lattice, and converts the enumeration problem into an optimization of prefix determinants.

\textit{Controlling adaptive HKZ projections.} Successive projected lattices in HKZ reduction depend on previously selected shortest vectors. To control these choices, we average over all admissible primitive vectors and establish an exponential-moment bound for the cumulative deficits of projected shortest-vector lengths from geometric upper bounds. This yields the predicted bulk Gram–Schmidt geometry and a uniform prefix profile, simultaneously for all quasi-HKZ bases.

\textit{Dense sublattices with suitable blocks.} First moments locate the relevant determinant scale, but overlapping candidates can make the untruncated second moment diverge. We exclude candidates containing a proper sublattice below an explicit determinant cutoff. This controls contributions from intersecting pairs while retaining asymptotically the entire first moment. A marked averaging formula then allows us to require prescribed high-probability properties in both the normalized sublattice and its orthogonal quotient. Requiring these blocks to admit further splits makes any fixed finite binary construction possible.

\textit{Joint constraints across ranks.} The prefixes of a basis form one nested chain. Applying Kim’s chain formula to simultaneous determinant bounds gives an obstruction stronger than bounds at individual ranks. It yields both the area rule and the universal lower coefficient $(2-\sqrt2)/4$, whose geometric interpretation follows by reducing the outer blocks.

\textit{The conditional plateau construction.} For the heuristic upper bound, we prescribe a coarse chain with $2$ outer blocks and $K - 2$ middle blocks. We verify that every nonempty subchain has a strictly positive first-moment exponent. Under the stated existence assumption, deterministic block reductions fill the gaps between prescribed prefixes, introducing peaks of height $O(1/K)$ above the critical plateau. Taking sufficiently many blocks therefore approaches the universal lower bound.

\paragraph{Organization.}
Section~\ref{sec:preliminaries} collects existing geometric and averaging results. Section~\ref{sec:tree-size} relates actual enumeration widths and tree sizes to prefix determinants. Section~\ref{sec:profiles} proves the matching enumeration bounds for quasi-HKZ bases. Section~\ref{sec:splitting} gives the two-block construction and a general framework for recursive splits. Section~\ref{sec:heuristic-upper} describes the construction attaining the lower bound under a suitable heuristic assumption. Section~\ref{sec:bounds} proves the lower bound and explains its interpretation in terms of profile area. Section~\ref{sec:fgl-plateau-note} finally describes the extension to the collision search method of Feng--Gao--Liu.


\subsection*{AI usage}
OpenAI's ChatGPT was used in preparing this note, including literature searches, the development and checking of proof arguments, and drafting and editing the text and \LaTeX{} source. The author is responsible for the mathematical claims, references, and final presentation.


\section{Preliminaries}
\label{sec:preliminaries}

We first introduce the lattice notation and averaging formulas used in the analysis. All logarithms are natural. In the algorithmic bounds, we suppress polynomial factors in the bit length of the encoded bases and targets. The Haar model describes the geometric distribution of the lattices.


\subsection{Lattices, projections and reduction}

A full-rank lattice is $\mathcal{L} = \mathbf{B}\Z^n \subset \R^n$, with covolume $\det \mathcal{L} = |\det \mathbf{B}|$ and minimum $\lambda_1(\mathcal{L}) = \min_{\boldsymbol{v} \in \mathcal{L} \setminus \{\boldsymbol{0}\}} \norm{\boldsymbol{v}}$. A lattice is \emph{unimodular} if $\det \mathcal{L} = 1$. For a sublattice $\mathcal{M}$, $\det \mathcal{M}$ is its covolume in its span. It is \emph{primitive} if a basis of $\mathcal{M}$ can be extended to a basis of $\mathcal{L}$. Equivalently, it is primitive if $\mathcal{M} = \mathcal{L} \cap \operatorname{span}_{\R} \mathcal{M}$. A vector $\boldsymbol v$ is primitive if $\Z\boldsymbol v$ is primitive. The quotient $\mathcal X_n = \mathrm{SL}_n(\R)/\mathrm{SL}_n(\Z)$ is the space of unimodular lattices, and $\mu_n$ is its Haar probability measure. When referring to a Haar-random lattice, we mean $\mathcal{L} \sim \mu_n$. The quantities studied in this paper are invariant under orthogonal transformations, and all choices of bases and shortest vectors are understood to be measurable.

For an ordered column basis $\mathbf{B} = (\boldsymbol{b}_1, \ldots, \boldsymbol{b}_n)$, let $\boldsymbol{b}_i^*$ denote its Gram--Schmidt vectors, and write
\begin{align}
\boldsymbol{b}_j = \boldsymbol{b}_j^* + \sum_{i<j} \mu_{j,i} \boldsymbol{b}_i^*, \qquad a_i = \norm{\boldsymbol{b}_i^*}, \qquad y_i = \ln a_i.
\end{align}
For unimodular lattices $\mathcal{L}$, $\sum_i y_i = 0$. For a subspace $W$, write $\pi_W$ for orthogonal projection onto $W$. For a fixed basis, we also write $\pi_i$ for projection onto $\operatorname{span}(\boldsymbol{b}_1, \ldots, \boldsymbol{b}_{i-1})^\perp$. The remaining dimension is then $d = n - i + 1$, and the projected lattice satisfies:
\begin{align}\label{eq:projected}
\mathcal{L}_{\mathbf{B},d} = \pi_i(\mathcal{L}), \qquad \det\left(\mathcal{L}_{\mathbf{B},d}\right) = \prod_{j=i}^n a_j.
\end{align}
Projection removes the components along the first $i-1$ basis vectors: the images of $\boldsymbol b_i,\ldots,\boldsymbol b_n$ generate the remaining rank-$d$ lattice. We identify this orthogonal complement with $\R^d$. The normalized basis depth is $x_i = (i - 1) / n = 1 - d / n$. The figures place the $i$th Gram--Schmidt coordinate at $i/n$; this shift of $1/n$ has no effect on the limiting curves. Throughout the paper we use the following notation for the log-determinant of the sublattice generated by $\boldsymbol{b}_1,\ldots,\boldsymbol{b}_k$:
\begin{align}
P_k(\mathbf{B}) = \sum_{i=1}^k \ln a_i \qquad (0 \leq k \leq n).
\end{align}
For lattices with $\det \mathcal{L} = 1$, this gives $P_0(\mathbf{B}) = P_n(\mathbf{B}) = 0$ and $\ln \det \mathcal{L}_{\mathbf{B},d} = -P_{n-d}(\mathbf{B})$. For unimodular lattices, $P_k(\mathbf{B})$ coincides with the Rankin profile studied by Li--Nguyen~\cite{LiNguyen2025}. This also relates to Neumaier's bit profile~\cite{Neumaier2017}, defined using leading Gram determinants, which has coordinates $2 P_k(\mathbf{B})/\ln 2$.

A basis is called \emph{size-reduced} if $|\mu_{j,i}| \leq 1/2$ for all $i < j$. It is \emph{HKZ-reduced} if it is size-reduced and $a_i = \lambda_1(\pi_i(\mathcal{L}))$ for every $i$. For $n \geq 2$, a basis is \emph{quasi-HKZ-reduced}~\cite[Definition~2.4]{hscomplete} if it is size-reduced, the tail $(\pi_2(\boldsymbol{b}_2), \ldots, \pi_2(\boldsymbol{b}_n))$ is HKZ-reduced, and
\begin{align}\label{eq:lll-pair}
\frac{3}{4} a_1^2 \leq a_2^2 + \mu_{2,1}^2 a_1^2.
\end{align}
This is the usual LLL condition~\cite[(1.4)--(1.5)]{lll} on the first pair $(\boldsymbol{b}_1, \boldsymbol{b}_2)$. In particular, $a_1 \leq \sqrt2 a_2$. Every HKZ basis is quasi-HKZ: comparing $\boldsymbol{b}_1$ with $\boldsymbol{b}_2$ gives the stronger inequality $a_1^2 \leq a_2^2 + \mu_{2,1}^2 a_1^2$.


\subsection{Geometric bounds and averaging}

Let $\mathbb{B}_d(R)$ denote the Euclidean ball of radius $R$, and let $\Gamma$ and $\zeta$ denote the gamma and Riemann zeta functions. We will use the following constants:
\begin{align}
V_d = \frac{\pi^{d/2}}{\Gamma(1 + d/2)} \, , \qquad r_d = V_d^{-1/d} \quad (d \geq 1), \qquad
\kappa_d = \frac{V_d}{2\zeta(d)} \quad (d \geq 2).
\end{align}
Here $V_d$ is the volume of the unit ball in $d$ dimensions, and $r_d$ is the radius of a ball of volume one.

\begin{lemma}[Geometric bounds]
\label{lem:geometry}
For every rank-$d$ lattice $\mathcal{L}$, the following inequalities hold, with the asymptotic estimate taken as $d\to\infty$:
\begin{align}\label{eq:minkowski-stirling}
\frac{\lambda_1(\mathcal L)}{(\det\mathcal L)^{1/d}}\leq2r_d\leq\sqrt d,\qquad \ln r_d=\frac{1}{2}\ln\frac{d}{2\pi e}+o(1).
\end{align}
\end{lemma}
\begin{proof}
Minkowski's convex body theorem~\cite[Theorem~4]{nguyenhermite} gives the first inequality; the cube of side $2/\sqrt d$ inside the unit ball gives the second. Stirling's formula gives the asymptotic and $\ln V_d=-d\ln d/2+O(d)$.
\qed\end{proof}

For Haar-random unimodular lattices, the sharper estimate $\lambda_1(\mathcal L)=(1+o(1))r_d$ holds with probability $1-o(1)$~\cite[Theorem~1]{sodergren}. We will also need the deterministic bound: the projections arising later in the analysis depend on the chosen basis, and need not be Haar-random. A standard consequence of Siegel's primitive-vector mean~\cite{siegel} is the explicit lower-tail estimate
\begin{align}\label{eq:short-tail}
 \PP_{\mu_d}\{\lambda_1(\mathcal L)<r_d/2\}\leq\kappa_d(r_d/2)^d=\frac{2^{-d}}{2\zeta(d)}.
\end{align}
Indeed, the expected number of primitive pairs $\{\pm\boldsymbol v\}$ in that ball is its volume divided by $2\zeta(d)$, and every shortest vector is primitive.

For a primitive pair $[\boldsymbol v]=\{\pm\boldsymbol v\}$ in a unimodular rank-$d$ lattice, we obtain a smaller unimodular lattice by projecting away the line spanned by $\boldsymbol v$ and rescaling the result to determinant one:
\begin{align}\label{eq:vector-quotient}
 \mathcal Q_{\boldsymbol v}=\|\boldsymbol v\|^{1/(d-1)}\pi_{\boldsymbol v^\perp}(\mathcal L)\in\mathcal X_{d-1}.
\end{align}
The projected lattice consists of the components of lattice vectors perpendicular to $\boldsymbol v$. Its determinant is $1/\|\boldsymbol v\|$, which explains the scaling factor. This is exactly the smaller lattice encountered after one HKZ step, normalized to unit determinant.

\begin{proposition}[Averaging over primitive vectors]\label{prop:unfold}
Let $d\geq2$ and let $f(s,\mathcal Q)$ be nonnegative, measurable and rotation-invariant in $\mathcal Q$. Then
\begin{align}\label{eq:unfold-measure}
 \E_{\mu_d}\!\left[\sum_{[\boldsymbol v]\in\mathcal L_{\rm prim}/\{\pm1\}}f(\|\boldsymbol v\|,\mathcal Q_{\boldsymbol v})\right]
 =d\kappa_d\int_0^\infty s^{d-1}\E_{\mu_{d-1}}[f(s,\mathcal Q)]\dd s.
\end{align}
Here $\mathcal X_1=\{\Z\}$ and $\mu_1$ is its point mass.
\end{proposition}
\begin{proof}
We apply the primitive component ($k=1$) of Str\"ombergsson's averaging formula~\cite[Proposition~5.1 and (5.1)--(5.4)]{strom}. For a specified primitive vector of length $s$, the normalized quotient carries Haar measure $\mu_{d-1}$, while the choices of lifts have total measure one. Counting $\boldsymbol v$ and $-\boldsymbol v$ separately gives primitive-vector intensity $\zeta(d)^{-1}\dd\boldsymbol v$. Passing to polar coordinates and dividing by two to count primitive pairs gives the factor $d\kappa_d$.
\qed\end{proof}

In Section~\ref{sec:hkz-moment}, we use this with $f(s,\mathcal Q)=\ind_{s\leq2r_d}(2r_d/s)^{d/2}\exp(Z_{d-1}(\mathcal Q)/2)$, where $Z_{d-1}$ measures cumulative HKZ length deficits. We also use the same quotient weight with the simpler cutoff $\ind_{s\leq4r_d}$ for quasi-HKZ bases. Summing over all eligible first vectors will allow us to obtain a bound that holds for every possible first HKZ step.


\subsection{Duality, Gaussian mass, and sublattice averages}
\label{sec:mean-inputs}

The dual of a lattice $\mathcal L$ in its span is $\mathcal L^*=\{\boldsymbol y:\langle\boldsymbol y,\mathcal L\rangle\subset\Z\}$. Haar measure is invariant under duality. The Gaussian mass and its basic identities are
\begin{align}\label{eq:theta-input}
 \rho(\mathcal L)=\sum_{\boldsymbol v\in\mathcal L}e^{-\pi\|\boldsymbol v\|^2},\qquad
 \rho(\mathcal L)=\frac{\rho(\mathcal L^*)}{\det\mathcal L},\qquad
 \E_{\mu_n}[\rho(\mathcal L)]=2\quad(n\geq2).
\end{align}
The middle identity is Poisson summation; the last is Siegel's mean-value formula~\cite{siegel}, including the zero vector.

Following Kim, write $\Gr(\mathcal L,r)$ for the primitive rank-$r$ sublattices, counted without orientations or choices of bases. The notation comes from the Grassmannian: each such sublattice is determined by its span through $\mathcal M=\mathcal L\cap\operatorname{span}_{\R}\mathcal M$. The coefficient $a(n,r)$ is the expected number of such sublattices with determinant at most one:
\begin{align}\label{eq:sublattice-constant}
 a(n,r)=\frac{1}{n}\binom{n}{r}\prod_{j=1}^r\frac{V_{n-j+1}\zeta(j)}{V_j\zeta(n-j+1)},\qquad \zeta(1):=1.
\end{align}
Kim's mean-value formula~\cite[Theorem~3, $k=1$]{kimflags} states that
\begin{align}\label{eq:sublattice-mean}
 N_{r,H}(\mathcal L)=\#\{\mathcal M\in\Gr(\mathcal L,r):\det\mathcal M\leq H\},\qquad
 \E_{\mu_n}[N_{r,H}(\mathcal L)]=a(n,r)H^n.
\end{align}
The mean-value formula for a fixed rank also follows from the classical averaging framework of Thunder~\cite{Thunder1998}. The $k=2$ case of the same theorem gives the mean number of ordered pairs with ranks $r,t$, determinant bounds $H,J$, and zero intersection: it equals $a(n,r)a(n,t)H^nJ^n$ when $r+t\leq n-1$. This identity concerns the average number of such pairs. In particular, it does not imply that sublattices chosen by an algorithm are independent.

The next lemma also keeps track of the two smaller lattices obtained by splitting along $\mathcal M$: the sublattice itself and its orthogonal quotient. After normalizing their determinants, their averages separate into two Haar integrals.

\begin{lemma}[Averaging the two blocks of a split]\label{lem:block-average}
For $\mathcal M\in\Gr(\mathcal L,r)$, let $D=\det\mathcal M$ and $\mathcal Q=\pi_{\mathcal M^\perp}(\mathcal L)$, so that $\det\mathcal Q=D^{-1}$. For every nonnegative measurable function $F$ that is rotation-invariant in both lattice arguments:
\begin{align}\label{eq:marked-disintegration}
 \E_{\mu_n}\!\left[\sum_{\mathcal M\in\Gr(\mathcal L,r)}F(D,D^{-1/r}\mathcal M,D^{1/(n-r)}\mathcal Q)\right] &=na(n,r)\int_0^\infty D^{n-1}\E_{\substack{\mathcal L_1\sim\mu_r\\\mathcal L_2\sim\mu_{n-r}}}[F(D,\mathcal L_1,\mathcal L_2)]\dd D,
\end{align}
where $\mathcal L_1$ and $\mathcal L_2$ on the right are independent.
\end{lemma}
\begin{proof}
Kim's decomposition of Haar measure~\cite[Lemma~10 and Sections~4--5]{kimflags} separates the determinant, the two normalized lattices, rotations, and choices of lifts. The determinant density is proportional to $D^{n-1}\dd D$, and the two lattice factors carry Haar probability measure. Integrating out the rotations and lifts leaves a constant; taking $F=\ind_{D\leq H}$ and using \eqref{eq:sublattice-mean} identifies it as $na(n,r)$.
\qed\end{proof}

For example, if the two blocks must belong to prescribed sets of Haar probabilities $p_1,p_2$, their expected count up to determinant $H$ is $a(n,r)H^np_1p_2$. This is an average over all candidate splits. A sublattice selected by an algorithm need not itself have a Haar distribution. Section~\ref{sec:marked-dense} uses the identity to exclude candidates with unsuitable blocks before choosing a split.


\section{Enumeration tree sizes}\label{sec:tree-size}

To compare different bases, we first relate their enumeration tree sizes to their prefix determinants. We show that, for random lattices, the largest prefix determinant determines the tree size up to a factor $e^{O(n)}$. This estimate holds simultaneously for all bases, and will therefore be used for both the upper and lower bounds. The connection between prefix determinants and enumeration costs also appears in the statistical study of Yu--Ducas~\cite[Section~5]{yuducas2018}, who analyzed the half-volume of BKZ-reduced bases as an indicator of enumeration costs. Here we use the prefix-determinant profile to control the full enumeration tree size.


\subsection{Levels, projections, and prefixes}\label{sec:enum-inputs}

At depth $d$, full enumeration has chosen the coefficients of the last $d$ basis vectors $\boldsymbol b_n,\boldsymbol b_{n-1},\ldots,\boldsymbol b_{n-d+1}$. The root corresponds to $d=0$, the leaves correspond to $d=n$, and the complementary prefix has rank $k=n-d$. The number of nodes at depth $d$ and the total tree size are given by:
\begin{align}\label{eq:tree}
 N_d(\mathbf{B},R)=\#\bigl(\mathcal{L}_{\mathbf{B},d}\cap\mathbb{B}_d(R)\bigr),\qquad
 T(\mathbf{B},R)=1+\sum_{d=1}^nN_d(\mathbf{B},R).
\end{align}
Each projected vector corresponds to one integer coefficient tuple. The root contributes one; sign symmetry changes the count by at most a constant factor. The fixed-radius full tree is independent of traversal order, with polynomial arithmetic overhead per node~\cite{finckepohst,schnorreuchner}.

For $\mathcal{M}_k=\Z\boldsymbol{b}_1+\cdots+\Z\boldsymbol{b}_k$ and $\mathcal{Q}_k=\pi_{\mathcal{M}_k^\perp}(\mathcal{L})$, primitivity of basis prefixes gives
\begin{align}\label{eq:prefix-projection}
 \det\mathcal{M}_k=e^{P_k(\mathbf{B})},\qquad
 \det\mathcal{Q}_k=e^{-P_k(\mathbf{B})},\qquad
 \mathcal{Q}_k^*=\mathcal{L}^*\cap\mathcal{M}_k^\perp.
\end{align}
A large prefix determinant therefore makes the corresponding projected lattice dense, which suggests that the enumeration tree will be wide at this depth. We make this relation precise below. In Section~\ref{sec:bounds}, we will also use that the prefixes form a nested chain of primitive sublattices.


\subsection{Uniform lower and upper estimates}\label{sec:prefix-nodes}

\begin{proposition}[Tree size from the prefix profile]\label{prop:tree-prefix}
With probability $1-e^{-\Omega(n)}$ over $\mathcal L\sim\mu_n$, simultaneously for all bases at radius $R=\lambda_1(\mathcal L)$,
\begin{align}\label{eq:tree-prefix}
 \ln T(\mathbf B)=\max_{0\leq k\leq n}P_k(\mathbf B)+O(n),\qquad
 \ln N_{n-k}(\mathbf B,R)=P_k(\mathbf B)+O(n)\quad(k<n).
\end{align}
\end{proposition}
\begin{proof}
We first prove the lower bound, with $d=n-k$. Averaging over a fundamental domain of $\mathcal Q_k$ shows that some translate of $\mathbb B_d(R/2)$ contains at least $V_d(R/2)^d/\det\mathcal Q_k$ lattice points. Subtracting one of these points from all of them gives distinct points in $\mathcal Q_k\cap\mathbb B_d(R)$, because distances within a ball of radius $R/2$ are at most $R$.

For the upper bound, every point with $\|\boldsymbol v\|\leq R$ contributes at least one to $e^{\pi R^2}\rho(\mathcal Q_k)$. Equivalently, the Gaussian function $e^{\pi(R^2-\|\boldsymbol v\|^2)}$ bounds the ball's indicator from above. Applying Poisson summation and then $\mathcal Q_k^*\subseteq\mathcal L^*$ gives
\begin{align}\label{eq:width-sandwich}
 e^{P_k}V_d(R/2)^d\leq N_d(\mathbf B,R)
 &\leq e^{\pi R^2}\rho(\mathcal Q_k)
 =e^{\pi R^2+P_k}\rho(\mathcal Q_k^*) \leq e^{\pi R^2+P_k}\rho(\mathcal L^*).
\end{align}
By \eqref{eq:short-tail} and Stirling's formula, with probability $1-e^{-\Omega(n)}$ we have $R\geq a\sqrt n$ for an absolute $a>0$. Markov's inequality and \eqref{eq:theta-input} also give $\rho(\mathcal L^*)\leq e^n$ with probability at least $1-2e^{-n}$. Both bounds therefore hold with probability $1-e^{-\Omega(n)}$. Using $R\leq\sqrt n$ from Lemma~\ref{lem:geometry} and the cube inside the unit ball,
\begin{align}
 V_d(R/2)^d\geq a^d(n/d)^{d/2}\geq e^{-O(n)},\qquad
 e^{\pi R^2}\rho(\mathcal L^*)\leq e^{(\pi+1)n}.
\end{align}
Taking logarithms gives the claimed estimate for each level of the tree. As there are only $n+1$ levels, the logarithm of the total tree size differs from that of the widest level by at most $\ln(n+1)$. This proves the proposition.
\qed\end{proof}

\paragraph{Shifted enumeration and CVPP.}
For a target $\boldsymbol t$, write $\boldsymbol t_k=\pi_{\mathcal M_k^\perp}(\boldsymbol t)$. In Poisson summation a shift only introduces complex factors of absolute value one, so the same upper bound holds uniformly over all targets:
\begin{align}\label{eq:shifted-width}
 \#\bigl((\mathcal Q_k-\boldsymbol t_k)\cap\mathbb B_{n-k}(R)\bigr)
 \leq e^{\pi R^2+P_k(\mathbf B)}\rho(\mathcal L^*).
\end{align}
Banaszczyk's covering-radius bound~\cite{banaszczyk}, together with \eqref{eq:short-tail} for $\mathcal L^*$, shows that radius $R=10\sqrt n$ contains a closest vector to every target with probability $1-e^{-\Omega(n)}$ for all sufficiently large $n$. With the bases constructed in Sections~\ref{sec:splitting} and~\ref{sec:heuristic-upper}, we can therefore solve exact CVPP in query time $\exp(\max_kP_k+O(n))$, giving the enumeration query bounds summarized in Table~\ref{tab:cvpp-complexities}. Depth-first traversal stores only the basis, target and one path, using polynomial space. Preprocessing is independent of the target and has unrestricted running time. It can search for a basis satisfying the required prefix bound; the heuristic bound remains conditional on Assumption~\ref{heur:assumption}. Section~\ref{sec:fgl-plateau-note} proves the smaller collision-search query bound of Theorem~\ref{thm:cvpp-collision}.

To justify the space bound, it remains to check that the basis retained after preprocessing has polynomial encoding length. On the event above, \eqref{eq:width-sandwich} and $N_{n-k}\geq1$ give $P_k\geq-O(n)$ for every prefix. Together with $\max_kP_k=O(n\ln n)$, this implies $|y_i|=O(n\ln n)$ and $\|\boldsymbol b_i\|\leq\exp(O(n\ln n))$ after size reduction. Relative to an encoded input basis, the integer change-of-basis matrix therefore has polynomial bit length. The resulting advice and query space are therefore polynomial, with the encoding factors understood as in Section~\ref{sec:preliminaries}.

In the rest of the paper, we will express prefix determinants through limiting profiles. If $P_{\lfloor xn\rfloor}/(n\ln n)=p(x)+o(1)$ uniformly in $x$, the enumeration tree coefficient is $\max_xp(x)$. The figures use prefix dimension $x=k/n$, so enumeration runs from the root at $x=1$ to the leaves at $x=0$.


\section{HKZ analysis}\label{sec:profiles}

We next analyze the enumeration trees of quasi-HKZ bases. We first derive a deterministic bound on HKZ prefix determinants, which will also be useful for the constructions in Section~\ref{sec:splitting}. To show that this bound is tight on random lattices, we then bound how far the successive HKZ lengths can fall below the Gaussian heuristic. This gives the limiting prefix profile and the enumeration exponent $1/(2e)$.


\subsection{HKZ prefix determinants}

Consider an HKZ basis $\mathbf B$ of a unimodular rank-$m$ lattice. At remaining dimension $j$, HKZ chooses a shortest vector in $\mathcal L_{\mathbf B,j}$; its length is the Gram--Schmidt length $a_{m-j+1}$. We divide by $(\det\mathcal L_{\mathbf B,j})^{1/j}$ to compare different scales, and let $h_j$ be the logarithm of this normalized length:
\begin{align}\label{eq:h-x}
 h_j=\ln\frac{\lambda_1(\mathcal{L}_{\mathbf{B},j})}{(\det\mathcal{L}_{\mathbf{B},j})^{1/j}}
 =y_{m-j+1}+\frac{P_{m-j}(\mathbf{B})}{j}.
\end{align}
The recursive Gaussian heuristic (GH) predicts $h_j\approx\ln r_j$. Substituting this prediction at successive ranks yields the heuristic HKZ Gram--Schmidt profile. The sandpile model of Hanrot--Pujol--Stehl\'{e}~\cite[Section~4.1]{hps} already yields the same leading profile; Haque--Pieprzyk~\cite[Sections~3.2 and~4]{haque} analyze a Gaussian-heuristic variant and derive more precise asymptotic estimates. Our aim is to justify the leading profile for random lattices despite the dependence between successive HKZ choices. Starting from $P_0=0$, the relation $P_{m-j+1}/(j-1)=P_{m-j}/j+h_j/(j-1)$ for $j\geq2$ gives the exact identity
\begin{align}\label{eq:prefix}
 P_{m-d}(\mathbf{B})=d\sum_{j=d+1}^m\frac{h_j}{j-1},\qquad1\leq d\leq m.
\end{align}

\begin{lemma}[Deterministic HKZ prefix bound]
\label{lem:hkz-bounds}
For every HKZ basis $\mathbf B_{\mathrm H}$ of a unimodular rank-$d$ lattice, uniformly for $0\leq t\leq1$,
\begin{align}\label{eq:hkz-bounds}
 \frac{P_{\lfloor td\rfloor}(\mathbf B_{\mathrm H})}{d\ln d}\leq p_{\mathrm H}(t)+o(1),\qquad\qquad
 p_{\mathrm H}(t)=-\frac{1-t}{2}\ln(1-t).
\end{align}
\end{lemma}

\begin{proof}
Minkowski's bound gives $h_j\leq\frac{1}{2}\ln j$. For $r=d-k\geq1$, \eqref{eq:prefix} implies
\begin{align}
 P_{d-r}\leq\frac{r}{2}\ln d\sum_{j=r+1}^d\frac{1}{j-1}
 \leq\frac{r}{2}\ln d\left(\ln(d/r)+\frac{1}{r}\right).
\end{align}
After dividing by $d\ln d$, the rounding error is uniform in $k$. At the remaining endpoint $k=d$, the claim follows from $P_d=0$.
\qed\end{proof}

In the split constructions, we will also use dual HKZ bases. Reversing and dualizing an HKZ basis of the dual lattice changes the prefix determinant $P_k$ to $P_{d-k}$. The corresponding upper bound is therefore the reflected profile:
\begin{align}\label{eq:dual-hkz-bound}
 p_{\mathrm D}(t)=p_{\mathrm H}(1-t)=-\frac{t}{2}\ln t.
\end{align}


\subsection{A one-sided moment bound}\label{sec:hkz-moment}

To compare the actual HKZ lengths with the recursive GH prediction, we measure their shortfall from $\ln(2r_j)$. We use this slightly larger reference value so that the deficits $e_j=\ln(2r_j)-h_j$ are nonnegative for every lattice. The additional $\ln2$ does not change the leading profile. For a unimodular rank-$m$ lattice $\mathcal L$, let $Z_m(\mathcal L)$ be the maximum of $\sum_{j=2}^m j e_j$ over its HKZ bases, and put $Z_1=0$. Thus $Z_m$ measures the total shortfall of the normalized logarithmic GS lengths from the recursive GH reference, weighted by the remaining dimension. The maximum is finite because each HKZ step has finitely many shortest-vector choices.

\begin{lemma}[A moment bound for the cumulative deficit]\label{lem:error}
For $\mathcal{L}\sim\mu_n$,
\begin{align}\label{eq:exp-moment}
 \E_{\mu_n}[e^{Z_n(\mathcal L)/2}]=e^{O(n^2)}.
\end{align}
\end{lemma}
\begin{proof}
We prove the claim by induction on the dimension, writing $M_d=\E_{\mu_d}[e^{Z_d/2}]$ and starting from $M_1=1$. For a fixed unimodular rank-$d$ lattice, let $\mathcal S(\mathcal L)$ be its finite nonempty set of shortest-vector pairs. Once a first vector is chosen, all later HKZ choices are choices in its projected quotient. Multiplying every lattice vector by the same factor changes its length and $(\det\mathcal L)^{1/j}$ by that factor, so $h_j$ and $e_j$ stay unchanged. We may therefore rescale the quotient to determinant one without changing the remaining deficits:
\begin{align}
 Z_d(\mathcal L)=\max_{[\boldsymbol v]\in\mathcal S(\mathcal L)}
 \left(d\ln\frac{2r_d}{\|\boldsymbol v\|}+Z_{d-1}(\mathcal Q_{\boldsymbol v})\right).
\end{align}
Every shortest vector is primitive and has length at most $2r_d$. After exponentiating, we can bound the maximum by the sum over all shortest-vector pairs, and then enlarge this sum to all primitive pairs in the ball of radius $2r_d$. All summands are nonnegative, so sums and integrals may be interchanged without first proving integrability (Tonelli's theorem~\cite[Theorem~2.37]{folland}). Proposition~\ref{prop:unfold} then gives:
\begin{align}\label{eq:moment-recursion}
 M_d\leq d\kappa_d M_{d-1}\int_0^{2r_d}s^{d-1}\left(\frac{2r_d}{s}\right)^{d/2}\dd s
 =\frac{2^d}{\zeta(d)}M_{d-1}\leq2^dM_{d-1}.
\end{align}
Here $V_dr_d^d=1$. The factor $1/2$ in the moment ensures convergence at $s=0$; removing it would give the divergent integral $\int_0^{2r_d}s^{-1}\dd s$. Iteration gives $\ln M_n\leq(\ln2)\sum_{d=2}^n d=O(n^2)$. Because we average over all primitive vectors, the argument does not require the quotients selected by HKZ to be Haar-random or independent.
\qed\end{proof}


\subsection{Quasi-HKZ profiles}\label{sec:quasi-profiles}

Every quasi-HKZ basis satisfies $\lambda_1(\mathcal{L})\leq a_1\leq\sqrt2\lambda_1(\mathcal{L})<4r_n$. Indeed, if the first vector is not shortest, a shortest vector has nonzero projection onto the HKZ tail, so $a_2\leq\lambda_1(\mathcal{L})$ and \eqref{eq:lll-pair} gives the claim. The last inequality follows from Minkowski's bound.

We now use the moment bound to justify the leading HKZ profile predicted by the heuristic models of Hanrot--Pujol--Stehl\'{e} and Haque--Pieprzyk~\cite{hps,haque}. The next proposition shows that the actual logarithmic Gram--Schmidt lengths $y_i$ differ from $\frac{1}{2}\ln n\,(1+\ln(1-x_i))$ by $o(\ln n)$ on average, away from the final ranks. Summing these coordinates gives the uniform prefix profile $p_{\mathrm H}$, whose maximum determines the tree exponent.

\begin{proposition}[Quasi-HKZ profiles on Haar-random lattices]\label{prop:profile}
With probability $1-e^{-\Omega(n)}$ over $\mathcal L\sim\mu_n$, simultaneously for all quasi-HKZ bases,
\begin{align}\label{eq:haar-prefix-profile}
 \max_{0\leq k\leq n}\left|\frac{P_k(\mathbf B)}{n\ln n}-p_{\mathrm H}(k/n)\right|=o(1).
\end{align}
With the same probability, for every fixed $0<\varepsilon<1$,
\begin{align}\label{eq:bulk-error}
 \frac{1}{n}\sum_{i:\,x_i\leq1-\varepsilon}
 \left|\frac{y_i}{\ln n}-\frac{1}{2}\bigl(1+\ln(1-x_i)\bigr)\right|=o(1).
\end{align}
\end{proposition}

\begin{proof}
Let $m=n-1$, and let $W_n$ be the maximum of $Z_m(\mathcal Q_{\boldsymbol v})$ over primitive pairs with $\|\boldsymbol v\|\leq4r_n$. Proposition~\ref{prop:unfold} and Lemma~\ref{lem:error} give
\begin{align}\label{eq:tail-moment}
 \E_{\mu_n}[e^{W_n/2}]\leq\kappa_n(4r_n)^n\E_{\mu_m}[e^{Z_m/2}]=e^{O(n^2)}.
\end{align}
Hence $W_n=O(n^2)$ with probability $1-e^{-\Omega(n^2)}$. Together with \eqref{eq:short-tail} and Proposition~\ref{prop:tree-prefix}, this gives an event of probability $1-e^{-\Omega(n)}$ on which all three bounds hold. In particular, every quasi-HKZ basis then has $|\ln a_1|=O(\ln n)$.

For its normalized HKZ tail $\mathbf B'$, put $g_j=\ln(2r_j)$ for $j\geq2$, with $g_1=e_1=0$. The prefix identity gives
\begin{align}\label{eq:reconstruct}
 y'_{m-d+1}=h_d-\sum_{j=d+1}^m\frac{h_j}{j-1},\qquad
 R_{m,d}=g_d-\sum_{j=d+1}^m\frac{g_j}{j-1}.
\end{align}
Using $h_j=g_j-e_j$ and the triangle inequality,
\begin{align}\label{eq:weighted-error}
 \sum_{d=1}^m d|y'_{m-d+1}-R_{m,d}|
 \leq\sum_{j=2}^m\left(j+\sum_{d<j}\frac{d}{j-1}\right)e_j
 \leq\frac{3}{2}W_n=O(n^2).
\end{align}
For fixed $\varepsilon>0$ and $d\geq\varepsilon m$, Stirling's formula and harmonic sums give $R_{m,d}=\frac{1}{2}\ln m\,(1+\ln(d/m))+O(1)$, where the constant may depend on $\varepsilon$. The weighted bound above therefore gives an average coordinate error $O(1)$ in this range. Rescaling by $y_i=y'_{i-1}-\ln a_1/m$ for $i\geq2$ and dividing by $\ln n$ proves \eqref{eq:bulk-error}.

Summing the coordinate errors and using $\frac{1}{2}\int_0^x(1+\ln(1-t))\dd t=p_{\mathrm H}(x)$ proves \eqref{eq:haar-prefix-profile} for $k/n\leq1-\varepsilon$.

For the remaining prefixes, Lemma~\ref{lem:hkz-bounds} applied to the tail gives $P_k/(n\ln n)\leq p_{\mathrm H}(k/n)+o(1)$, while Proposition~\ref{prop:tree-prefix} and $N_{n-k}\geq1$ give $P_k\geq-O(n)$. Since $P_n=0$ and $p_{\mathrm H}(x)\to0$ as $x\to1$, letting $\varepsilon\downarrow0$ proves the uniform claim.
\qed\end{proof}


\subsection{Matching enumeration bounds}\label{sec:hkz-enumeration}

\begin{proof}[Theorem~\ref{thm:enumeration}]
With probability $1-e^{-\Omega(n)}$, Propositions~\ref{prop:tree-prefix} and~\ref{prop:profile} both hold, giving uniformly over all quasi-HKZ bases:
\begin{align}\label{eq:quasi-critical}
 \frac{\ln T(\mathbf B)}{n\ln n}
 =\max_{0\leq x\leq1}p_{\mathrm H}(x)+o(1)
 =\frac{1}{2e}+o(1).
\end{align}
The maximum occurs at $x=1-1/e$.
\qed\end{proof}


\section{Splitting with dense sublattices}\label{sec:splitting}

To improve on HKZ, we construct a basis by first finding a dense primitive sublattice and then reducing this sublattice and its quotient separately. We show that, with high probability, a suitable split exists and that the same approach can be applied recursively. For the two-block construction, Proposition~\ref{prop:profile} gives the exact leading tree exponent. For the recursive construction, the deterministic upper bound from Lemma~\ref{lem:hkz-bounds} suffices.


\subsection{Dense sublattices with suitable blocks}\label{sec:marked-dense}

A primitive rank-$m$ sublattice $\mathcal M$ splits $\mathcal L$ into $\mathcal M$ and the rank-$(n-m)$ quotient $\mathcal Q=\pi_{\mathcal M^\perp}(\mathcal L)$. Their determinants are reciprocal. We need a split with small $\det\mathcal M$ whose two normalized blocks also have the usual high-probability properties of random lattices, such as the HKZ profile or the existence of further splits.

A lower bound at the relevant determinant scale follows from Shapira--Weiss~\cite[Proposition~3]{shapiraweiss2014}: for some
absolute constant $C>0$, with probability $1-o(1)$, every rank-$k$ sublattice $\mathcal M\subseteq\mathcal L$ satisfies
\begin{align}
\det\mathcal M \geq (n/C)^{k(n-k)/(2n)} \qquad (1\leq k<n).
\end{align}
In particular, at half rank this gives the lower scale $e^{-O(n)}n^{n/8}$. The following lemma provides a matching upper scale, up to $e^{O(n)}$ factors, while also ensuring suitable geometry of both normalized blocks.

\begin{lemma}[Dense sublattices with suitable blocks]\label{lem:marked-dense}
Fix $0<\alpha<1$. Take $m=\lfloor\alpha n\rfloor$ for $\alpha\ne1/2$, and $m=\lfloor(n-1)/2\rfloor$ for $\alpha=1/2$. There is a threshold
\begin{align}\label{eq:window}
 H=e^{O(n)}n^{m(n-m)/(2n)}
\end{align}
with the following property. Prescribe any rotation-invariant measurable properties that hold with Haar probability $1-o(1)$ in dimensions $m$ and $n-m$. With probability $1-o(1)$, $\mathcal L\sim\mu_n$ contains a primitive rank-$m$ sublattice $\mathcal M$ with $H/2\leq D=\det\mathcal M\leq H$ such that $D^{-1/m}\mathcal M$ and $D^{1/(n-m)}\mathcal Q$ satisfy the respective properties. Either property may instead concern the dual of that normalized block.
\end{lemma}

To prove the lemma, we count candidate splits and remove those containing an unusually dense proper sublattice. This removes only a small fraction of the candidates in expectation, while allowing us to control the correlations between intersecting candidates. The next two lemmas establish the required first- and second-moment bounds. We then use Lemma~\ref{lem:block-average} to remove candidates with unsuitable blocks, and show that at least one suitable split remains with high probability.

We first consider the case $2m\leq n-1$ and $m/n\to\alpha>0$. Let $N_H=N_{m,H}$ be the number of candidates and let $\mu=\E[N_H]=a(n,m)H^n$. We retain a candidate $\mathcal M$ only if every proper primitive sublattice $\mathcal C\subset\mathcal M$ satisfies:
\begin{align}\label{eq:internal-cutoff}
 \det\mathcal C\geq U_r:=n^{r(n-r)/(2n)},\qquad r=\rank\mathcal C\in\{1,\ldots,m-1\}.
\end{align}
Let $Z$ be the number of retained candidates. In the next two lemmas, we choose the positive constant implicit in the factor $e^{O(n)}$ in $H$ sufficiently large.

\begin{lemma}[Few candidates are removed]\label{lem:split-retention}
The reduced count $Z$ satisfies $\E[Z]=(1-e^{-\Omega(n)})\mu$.
\end{lemma}
\begin{proof}
For a candidate of determinant $D$, rescaling to determinant one changes the cutoff to $U_rD^{-r/m}$. By Lemma~\ref{lem:block-average}, the mean number of violations at rank $r$ in this normalized lattice is $a(m,r)U_r^mD^{-r}$. Summing over the ranks and integrating over $D$, we obtain:
\begin{align}\label{eq:retention-bound}
 \frac{\E[N_H-Z]}{\mu}
 &\leq\frac{na(n,m)}{\mu}\sum_{r=1}^{m-1}a(m,r)U_r^m\int_0^H D^{n-r-1}\dd D =\sum_{r=1}^{m-1}\frac{n}{n-r}a(m,r)U_r^mH^{-r}.
\end{align}
To estimate these terms uniformly in $r$, we write $b_j=\ln(V_j/\zeta(j))$, with $\zeta(1)=1$. Stirling's formula gives $b_j=-j\ln n/2+O(n)$ uniformly for $1\leq j\leq n$, since $j\ln(n/j)=O(n)$. The product for $a(m,r)$ has $r$ numerator and $r$ denominator terms, so
\begin{align}\label{eq:retention-coefficient}
 \ln a(m,r)=-\frac{r(m-r)}{2}\ln n+O(rn).
\end{align}
Consequently, the logarithm of the $r$th summand in \eqref{eq:retention-bound} is
\begin{align}\label{eq:retention-decay}
 -\frac{r(m-r)(n-m)}{2n}\ln n-r\ln\frac{H}{n^{m(n-m)/(2n)}}+O(rn).
\end{align}
Choosing the constant in $H$ sufficiently large makes this at most $-\kappa nr$ for some $\kappa>0$. Summing over $r$ proves the claim.
\qed\end{proof}

\begin{lemma}[The reduced count has small variance]\label{lem:split-pairs}
For the same choice of $H$, $\mu\geq e^{\Omega(n^2)}$ and $\E[Z^2]\leq(1+e^{-\Omega(n)})\mu^2$.
\end{lemma}
\begin{proof}
We distinguish pairs according to the rank of their intersection, and write $s=n-2m\geq1$. Pairs with zero intersection have mean $\mu^2$ before removing candidates, by Kim's pair formula in Section~\ref{sec:mean-inputs}. The diagonal contributes at most $\mu$. It remains to count distinct candidates intersecting in a rank-$r$ lattice $\mathcal C$, where $1\leq r<m$.

The intersection is primitive and has $D:=\det\mathcal C\geq U_r$. Projecting away $\mathcal C$ gives two primitive rank-$(m-r)$ sublattices with zero intersection in a rank-$(n-r)$ quotient of determinant $D^{-1}$. Each projected sublattice determines its original candidate as its full inverse image, so there is no additional multiplicity. Normalizing the quotient changes the two determinant ceilings to $HD^{-(n-m)/(n-r)}$.

Kim's pair formula applies because $2(m-r)\leq n-r-1$. Integrating its value against the intersection density $na(n,r)D^{n-1}\dd D$ bounds the number of these pairs by
\begin{align}\label{eq:overlap-integral}
 &na(n,r)a(n-r,m-r)^2H^{2(n-r)}\int_{U_r}^{\infty}D^{2m-n-1}\dd D =\frac{n}{s}a(n,r)a(n-r,m-r)^2H^{2(n-r)}U_r^{-s}.
\end{align}
The cutoff removes the divergence at zero, while $s\geq1$ ensures convergence at infinity. This is why, for $\alpha=1/2$, we chose $m=\lfloor(n-1)/2\rfloor$ rather than $\lfloor n/2\rfloor$.

Combining the diagonal, the pairs with zero intersection, and the remaining intersecting pairs gives:
\begin{align}\label{eq:variance}
 \frac{\E[Z^2]}{\mu^2}\leq1+\frac{1}{\mu}
 +\frac{n}{s}\sum_{r=1}^{m-1}\frac{a(n,r)a(n-r,m-r)^2}{a(n,m)^2}H^{-2r}U_r^{-s}.
\end{align}
Using the same $b_j$ as above, cancellation in the products for $a(n,r)$ gives
\begin{align}\label{eq:overlap-coefficient}
 \ln\frac{a(n,r)a(n-r,m-r)^2}{a(n,m)^2}
 &=2\sum_{j=m-r+1}^m b_j-\sum_{j=n-r+1}^n b_j-\sum_{j=1}^r b_j+O(r\ln n)\\
 &=\frac{r(s+r)}{2}\ln n+O(rn).
\end{align}
The logarithm of the $r$th contribution to \eqref{eq:variance} is therefore
\begin{align}\label{eq:overlapdecay}
 -\frac{r(m-r)(n-m)}{n}\ln n-2r\ln\frac{H}{n^{m(n-m)/(2n)}}+O(rn).
\end{align}
The error term also absorbs the factor $n/s\leq n$. As in Lemma~\ref{lem:split-retention}, increasing the constant in $H$ makes the sum $e^{-\Omega(n)}$. Finally, $\ln a(n,m)=-m(n-m)\ln n/2+O(n^2)$, so the same choice makes $\ln\mu=\Omega(n^2)$.
\qed\end{proof}

\begin{proof}[Lemma~\ref{lem:marked-dense}]
We first assume $2m\leq n-1$. Restricting to candidates with determinant at least $H/2$ removes at most $2^{-n}\mu$ candidates in expectation. If the two prescribed properties have failure probabilities $\varepsilon_m$ and $\varepsilon_{n-m}$, Lemma~\ref{lem:block-average} bounds the expected number of candidates with an unsuitable block by $(\varepsilon_m+\varepsilon_{n-m})\mu=o(\mu)$. These estimates also apply after imposing the internal cutoff, since doing so can only decrease the number of unsuitable candidates. Let $Y$ count the candidates satisfying the cutoff, the determinant window and both block properties. The two preceding lemmas give
\begin{align}\label{eq:split-existence}
 \E[Y]=(1-o(1))\mu,\qquad
 \PP\{Y>0\}\geq\frac{\E[Y]^2}{\E[Y^2]}\geq\frac{\E[Y]^2}{\E[Z^2]}=1-o(1).
\end{align}
Here the probability estimate follows by applying Cauchy--Schwarz to $Y\ind_{Y>0}$.

To cover the case $\alpha>1/2$, we apply the result at rank $n-m$ in the dual lattice $\mathcal L^*$. If $\mathcal N\subset\mathcal L^*$ is the resulting primitive sublattice, take $\mathcal M=\mathcal L\cap\mathcal N^\perp$. Then $\det\mathcal M=\det\mathcal N$, while $\pi_{\mathcal M^\perp}(\mathcal L)=\mathcal N^*$ and $\pi_{\mathcal N^\perp}(\mathcal L^*)=\mathcal M^*$. The two normalized blocks are therefore swapped and dualized. Haar invariance under duality transfers the prescribed properties, proving the remaining rank fractions and the dual variants.
\qed\end{proof}

The resulting determinant exponent is optimal to leading order. Indeed, for every fixed $\varepsilon>0$, the first-moment bound gives:
\begin{align}
 \PP\left\{\exists\mathcal M\in\Gr(\mathcal L,m):\det\mathcal M\leq n^{m(n-m)/(2n)-\varepsilon n}\right\}
 \leq n^{-\varepsilon n^2+O(n^2/\ln n)}=o(1).
\end{align}
The exponent $m(n-m)/(2n)$ is maximized at $m=n/2$, where it is $n/8$. Thus the half-rank construction has the scale $n^{n/8+o(n)}$ used in the introduction.


\subsection{The explicit two-block split}\label{sec:two-way}

We now apply the lemma to a sublattice of approximately half the ambient rank, fixing the middle prefix determinant at $n^{n/8+o(n)}$. We complete the left block with a reversed dual-HKZ basis and the right block with an HKZ basis. Requiring the normalized blocks to have the usual Haar profiles allows us to determine the entire limiting prefix profile. Its two maxima occur inside the blocks, on either side of the split.

\begin{proof}[Theorem~\ref{thm:two-way}]
Set $m=\lfloor(n-1)/2\rfloor$. Lemma~\ref{lem:marked-dense} supplies a primitive rank-$m$ sublattice $\mathcal{M}$ with $H/2\leq D=\det\mathcal{M}\leq H$, where $H=e^{O(n)}n^{m(n-m)/(2n)}$. In particular, $\ln D/(n\ln n)=1/8+o(1)$. Let $\mathcal{Q}=\pi_{\mathcal{M}^\perp}(\mathcal{L})$, so that $\det\mathcal{Q}=D^{-1}$.

When applying Lemma~\ref{lem:marked-dense}, we require the profile conclusions of Proposition~\ref{prop:profile} for the normalized lattices $D^{1/m}\mathcal M^*$ and $D^{1/(n-m)}\mathcal Q$. For the coordinate and first-vector claims below, we also require $r_d/2\leq\lambda_1\leq2r_d$ in these two lattices and in $D^{-1/m}\mathcal M$. These properties hold with probability $1-o(1)$ by \eqref{eq:short-tail} and Haar invariance under duality.

Let $\mathbf{C}$ and $\mathbf{E}$ be HKZ bases of these two normalized lattices, respectively. We undo both scalings, dualize and reverse the first basis, and append arbitrary lifts of the second basis. Since $\mathcal M$ is primitive, the resulting concatenation is a basis $\mathbf B$ of $\mathcal L$. We may size-reduce this basis without changing its Gram--Schmidt lengths. The prefix determinants obey the exact identities
\begin{align}\label{eq:split-prefix-identities}
  P_k(\mathbf{B})&=\frac{k}{m}\ln D+P_{m-k}(\mathbf{C}),&&0\leq k\leq m,\\
 P_{m+j}(\mathbf{B})&=\left(1-\frac{j}{n-m}\right)\ln D+P_j(\mathbf{E}),&&0\leq j\leq n-m.
\end{align}
We now apply the prescribed profile estimates to $\mathbf C$ and $\mathbf E$. As $m/n\to1/2$, the logarithms of both block dimensions are $\ln n+O(1)$. Together with $\ln D/(n\ln n)=1/8+o(1)$, this gives, uniformly for $0\leq x\leq1$:
\begin{align}\label{eq:two-way-profile}
 \frac{P_{\lfloor xn\rfloor}(\mathbf{B})}{n\ln n}=p_2(x)+o(1),\qquad
 p_2(x)=\frac{v}{8}-\frac{v}{4}\ln v,\quad v=2\min\{x,1-x\}.
\end{align}
The two maxima occur at $x=1/(2\sqrt e)$ and $x=1-1/(2\sqrt e)$, where $p_2(x)=1/(4\sqrt e)$. With probability $1-o(1)$, this basis exists and Proposition~\ref{prop:tree-prefix} applies to every basis. Its normalized tree exponent is therefore $\max_xp_2(x)+o(1)$, proving \eqref{eq:two-way-main}.
\qed\end{proof}

\paragraph{The Gram--Schmidt profile and the first vector.}\label{sec:split-gs}
Proposition~\ref{prop:profile}, with reflection and sign reversal in the first block, gives the curve in Figure~\ref{fig:gs-profiles}:
\begin{align}\label{eq:split-gs-profile}
 g_2(x)=p_2'(x)=
 \begin{cases}
 -\frac{1}{4}-\frac{1}{2}\ln(2x),&0<x<\frac{1}{2},\\
 \frac{1}{4}+\frac{1}{2}\ln(2(1-x)),&\frac{1}{2}<x<1.
 \end{cases}
\end{align}
The normalized GS lengths converge to this curve in averaged absolute error on every closed subinterval contained in either $(0,1/2)$ or $(1/2,1)$. At the split, the first lengths $c_1,e_1=n^{1/2+o(1)}$ of $\mathbf C,\mathbf E$ give
\begin{align}\label{eq:split-gs-jump}
 a_m=\frac{D^{1/m}}{c_1}=n^{-1/4+o(1)},\qquad
 a_{m+1}=D^{-1/(n-m)}e_1=n^{1/4+o(1)}.
\end{align}
The ratio of the two Gram--Schmidt lengths at the split is therefore $n^{1/2+o(1)}$. For the first basis vector, $\lambda_1(\mathcal L)\leq\sqrt n$ and the bound on $\lambda_1(D^{-1/m}\mathcal M)$ give
\begin{align}\label{eq:split-first-vector}
 \|\boldsymbol b_1\|\geq\lambda_1(\mathcal M)
 \geq D^{1/m}\frac{r_m}{2}=n^{3/4+o(1)},\qquad
 \frac{\|\boldsymbol b_1\|}{\lambda_1(\mathcal L)}\geq n^{1/4-o(1)}.
\end{align}
A stronger prediction for the first vector follows if we assume recursive GH, in the form $h_j=\frac{1}{2}\ln j+O(1)$ at every rank of $\mathbf C$. Under this additional assumption, \eqref{eq:reconstruct} gives:
\begin{align}\label{eq:split-first-vector-gh}
 \ln\|\boldsymbol b_1\|
 =\frac{\ln D}{m}+\sum_{j=2}^m\frac{h_j}{j-1}
 =\frac{1}{4}(\ln n)^2+O(\ln n).
\end{align}
The lower bound in \eqref{eq:split-first-vector} does not require this heuristic assumption.


\subsection{Recursive splitting}
\label{sec:recursive-splits}

The approach of finding a dense sublattice, splitting into blocks, and reducing each block separately, can be applied recursively. This recursive idea was already described by Hanrot--Stehl\'{e}~\cite{hscomplete}, and recently studied by~\cite{aeps2025}. For any fixed finite binary tree with split fractions in $(0,1)$, Lemma~\ref{lem:marked-dense} gives the required splits with probability $1-o(1)$. To see this, start with HKZ or reversed dual-HKZ bases at the leaves, and inductively require each normalized child lattice to admit the construction specified by its subtree.

For a split at fraction $\alpha$, let $q_\alpha=\alpha(1-\alpha)/2$. The upper profiles of the two child lattices then combine as follows:
\begin{align}\label{eq:block-composition}
 p(\alpha t)&=tq_\alpha+\alpha p_{\mathrm L}(t),\\
 p(\alpha+(1-\alpha)t)&=(1-t)q_\alpha+(1-\alpha)p_{\mathrm R}(t),\qquad 0\leq t\leq1.
\end{align}
For a fixed tree, rounding and determinant slack contribute $o(n\ln n)$, so Proposition~\ref{prop:tree-prefix} gives $T(\mathbf B)\leq n^{(\max p)n+o(n)}$.

To find trees with a small maximum prefix determinant, the accompanying script \texttt{recursive\_binary\_splits.py} performs the following search:
\begin{enumerate}
\item At depth $0$, we tabulate the maxima of $p_{\mathrm H}(t)+st$ and $p_{\mathrm D}(t)+st$.
\item At each further depth, we balance the two child peaks using \eqref{eq:block-composition}, with inherited tilts $s+(1-\alpha)/2$ and $s-\alpha/2$. We retain an earlier candidate if it gives a smaller maximum.
\item We reconstruct the finite tree and evaluate the maximum in each terminal block separately with high-precision arithmetic.
\end{enumerate}

\begin{proof}[Theorem~\ref{thm:eight-blocks}, with numerical evaluation]
The numerical search first gives a coefficient below $0.1474$ at depth $13$. To leave a larger margin, we use its depth-$14$ candidate with $16\,384$ blocks. Evaluating the maximum in each terminal block separately gives:
\begin{align}\label{eq:recursive-numerical-bound}
 \max_xp(x)=0.147399340499\ldots<0.14739935<0.1474.
\end{align}
The finite-tree construction above then proves the claimed bound as $n\to\infty$. The script and recorded trees accompany the article. The proof only uses this specific candidate, and does not require the numerical search to find an optimal tree. The $128$-block curve in Figure~\ref{fig:tree-profiles} is the depth-$7$ example, and Figure~\ref{fig:stack-provable} shows depths $0$ through $10$.
\qed\end{proof}

\begin{figure}[!t]
\centering
\includegraphics[width=\linewidth]{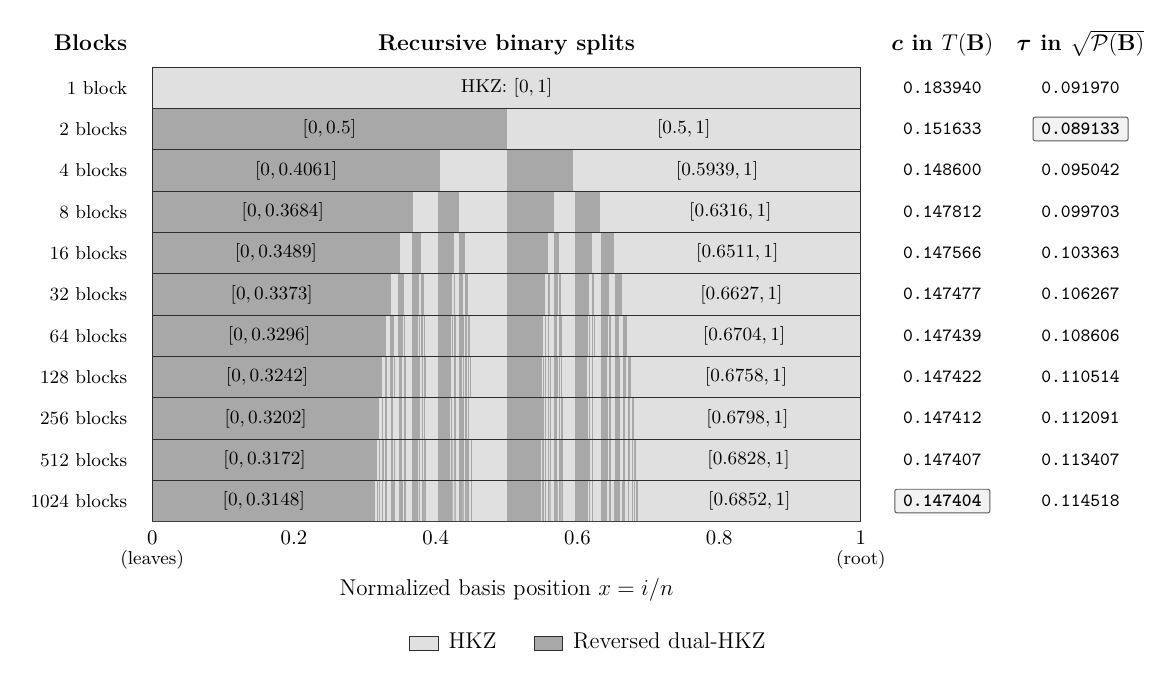}
\caption{Recursive binary splits for various recursion depths. As the depth increases, the computed upper bound on $c$ stabilizes around $0.147397\ldots < 0.1474$. The same construction does not lower the FGL collision search cost constant $\tau$ beyond the binary split.}
\label{fig:stack-provable}
\end{figure}


\section{Heuristic optimal bases}
\label{sec:heuristic-upper}

To approach the lower bound $(2-\sqrt2)/4$, we use a different arrangement of blocks. The two outer blocks have fixed normalized width, while the remaining interval is divided into $K-2$ equally sized middle blocks. As $K$ increases, the middle blocks become smaller. Under Assumption~\ref{heur:assumption}, this construction gives:
\begin{align}\label{eq:heuristic-block-upper}
 c_K\leq\frac{2-\sqrt2}{4}\left(1+\frac{\sqrt2}{e(K-2)}\right),\qquad K\geq3.
\end{align}
The proof of Theorem~\ref{thm:heuristic} has two steps. We first verify that the prescribed chain satisfies the first-moment conditions in the assumption. We then complete this chain to a basis using HKZ and reversed dual-HKZ reduction inside the blocks.

\paragraph{Precise form of Assumption~\ref{heur:assumption}.}
Fix $r\geq2$, rank fractions $0=x_0<\cdots<x_r=1$, and determinant exponents $h_0=h_r=0,h_1,\ldots,h_{r-1}$, independently of $n$. Let $C_n(\mathcal L)$ count chains of primitive sublattices satisfying
\begin{align}\label{eq:heuristic-chain}
 \begin{gathered}
 \{0\}=\mathcal M_0\subset\cdots\subset\mathcal M_r=\mathcal L, \qquad \rank\mathcal M_i=\lfloor x_i n\rfloor,\qquad \det\mathcal M_i\leq n^{h_i n}\quad(1\leq i<r).
 \end{gathered}
\end{align}
For nonempty $J\subseteq\{1,\ldots,r-1\}$, let $C_{J,n}(\mathcal L)$ count chains satisfying just the ranks and bounds indexed by $J$, whether or not they extend to a full chain. The assumption says that, if some fixed $\delta>0$ satisfies
\begin{align}\label{eq:heuristic-moment-condition}
 \min_{\varnothing\ne J\subseteq\{1,\ldots,r-1\}}\E_{\mu_n}[C_{J,n}(\mathcal L)]\geq n^{\delta n^2}
\end{align}
for all sufficiently large $n$, then $\PP\{C_n(\mathcal L)>0\}=1-o(1)$. The condition is imposed on every subchain because a large full-chain average alone may hide a rank at which suitable partial chains are unlikely to exist.

For our construction, let $c=(2-\sqrt2)/4$ and $a=2c=1-1/\sqrt2$. The two outer blocks each have width $a$, independently of $K$. We split the remaining interval of width $1-2a=\sqrt2-1$ into $K-2$ middle blocks, and impose the determinant bound $n^{cn}$ at every internal boundary:
\begin{align}\label{eq:heuristic-boundaries}
 w=\frac{1-2a}{K-2},\qquad x_i=a+(i-1)w\quad(1\leq i\leq K-1),\qquad
 \det\mathcal M_i\leq n^{cn}.
\end{align}

\begin{lemma}[Positive first moments for every subchain]\label{lem:heuristic-moments}
For fixed $K\geq3$, the ranks and bounds in \eqref{eq:heuristic-boundaries} satisfy \eqref{eq:heuristic-moment-condition} with $\delta=(1-2a)^2/(8(K-2))$.
\end{lemma}
\begin{proof}
Consider any nonempty set of boundaries $t_1<\cdots<t_s$. Kim's chain formula~\cite[Theorem~5]{kimflags} and Stirling's formula give $\E[C_{J,n}]=n^{(\Gamma_J+o(1))n^2}$. Using $c=a/2$ and $a-a^2/2=1/4$, the exponent becomes a sum of squares:
\begin{align}\label{eq:heuristic-square-slack}
 \Gamma_J&=c(1+t_s-t_1)-\frac{1-t_1^2-(1-t_s)^2-\sum_{\nu=1}^{s-1}(t_{\nu+1}-t_\nu)^2}{4}\\
 &=\frac{(t_1-a)^2+(1-a-t_s)^2+\sum_{\nu=1}^{s-1}(t_{\nu+1}-t_\nu)^2}{4}.
\end{align}
The distances in the last line are nonnegative integer multiples of $w$ whose sum is $(K-2)w$. Their squared sum is at least $(K-2)w^2$, since $j^2\geq j$ for every nonnegative integer $j$. Thus $\Gamma_J\geq(1-2a)^2/(4(K-2))=2\delta$. There are finitely many subchains for fixed $K$, so the same bound applies to all of them for sufficiently large $n$.
\qed\end{proof}

\begin{lemma}[Completing the blocks]\label{lem:heuristic-fill}
Any primitive chain with the ranks and determinant bounds in \eqref{eq:heuristic-boundaries} extends to a basis with
\begin{align}\label{eq:heuristic-prefix-ceiling}
 \max_j\frac{P_j(\mathbf B)}{n\ln n}\leq c+\frac{w}{2e}+o(1).
\end{align}
\end{lemma}
\begin{proof}
We reduce the first block by reversed dual-HKZ, the last block by HKZ, and each middle block by either orientation. Lifting and concatenating the resulting bases gives a basis of the full lattice. Within a block, the prefix profile is bounded by the linear interpolation of its two boundary heights plus its width times $p_{\mathrm H}$ or $p_{\mathrm D}$, by Lemma~\ref{lem:hkz-bounds}. Smaller boundary determinants only lower this interpolation.

The left outer block therefore has upper profile $x(1-\ln(x/a))/2$ for $0\leq x\leq a$, with maximum $a/2=c$; the right block is its reflection. Each middle block has both boundary heights at most $c$ and width $w$, so its peak is at most $c+w\max p_{\mathrm H}=c+w/(2e)$, up to $o(1)$. This proves the bound.
\qed\end{proof}

\begin{proof}[Theorem~\ref{thm:heuristic}]
Lemma~\ref{lem:heuristic-moments} and Assumption~\ref{heur:assumption} supply the chain with probability $1-o(1)$. Lemma~\ref{lem:heuristic-fill} and Proposition~\ref{prop:tree-prefix} then give \eqref{eq:heuristic-block-upper}, because
\begin{align}\label{eq:heuristic-peak}
 c+\frac{w}{2e}=\frac{2-\sqrt2}{4}\left(1+\frac{\sqrt2}{e(K-2)}\right).
\end{align}
For any fixed $\varepsilon>0$, we can choose $K$ independently of $n$ so that this bound is below $c+\varepsilon$. The resulting tree size is at most $n^{(c+\varepsilon)n+o(n)}$, proving \eqref{eq:heuristic}.

\qed\end{proof}

Figure~\ref{fig:stack-heuristic} shows this construction for total block counts $K=4,8,16,\ldots$, alternating the orientations of the middle blocks. For comparison, the one- and two-block rows show the HKZ and two-block constructions. As $K$ increases, the upper bound on the tree coefficient approaches the universal lower bound.

\begin{figure}[!t]
\centering
\includegraphics[width=\linewidth]{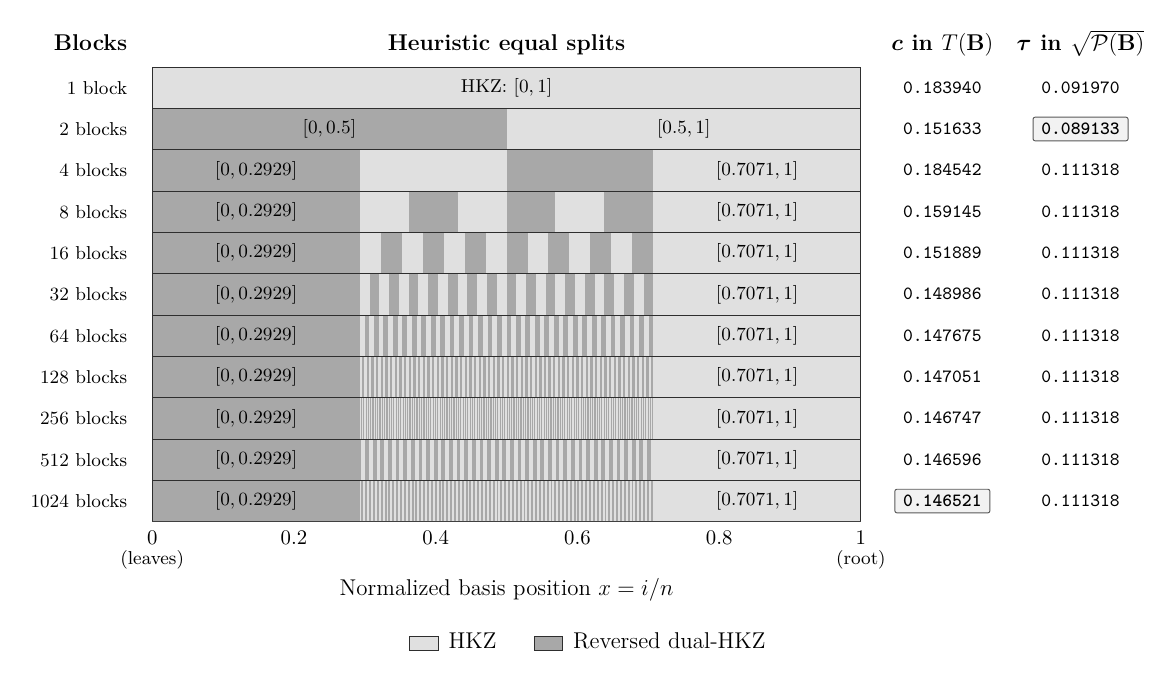}
\caption{Heuristic splits with equally sized middle blocks, using the same block counts as the recursive binary splits. The upper bound approaches the lower bound as $c_K = c_{\textrm{flat}} + O(1 / K)$. The same construction does not lower the FGL collision search cost constant $\tau$ beyond the binary split.}
\label{fig:stack-heuristic}
\end{figure}


\section{Lower bounds and profile area}
\label{sec:bounds}

To obtain a lower bound for all bases, we use that their primitive prefixes form nested chains. Kim's mean-value formula gives a lower bound on the area under a feasible prefix profile. By first reducing the two outer blocks, we turn this area constraint into the stronger lower bound on the maximum prefix determinant stated in Theorem~\ref{thm:any-basis}.


\subsection{The area rule}
\label{sec:profile-area}

A basis determines one primitive sublattice at every prefix rank. To rule out bases whose prefix profiles all lie below a fixed curve, we count these complete nested chains. For a continuous ceiling $h:[0,1]\to\R$, set $A(h)=\int_0^1h(x)\dd x$ and
\begin{align}\label{eq:area-chain-count}
 F_n(h;\mathcal L)=\#\left\{\mathcal M_1\subset\cdots\subset\mathcal M_{n-1}\subset\mathcal L:
 \begin{array}{l}\mathcal M_k\text{ primitive},\ \rank\mathcal M_k=k,\\
 \det\mathcal M_k\leq n^{nh(k/n)}\quad(1\leq k<n)\end{array}\right\}.
\end{align}
We count chains of sublattices, so different bases representing the same chain contribute only once. The expectation of $F_n$ is the average number of admissible chains in a Haar-random lattice. The next proposition shows that this average tends to zero if the area under $h$ is below $1/8$. In that case, the probability that any such chain exists also tends to zero.

\begin{proposition}[Area rule]\label{prop:area-rule}
For every fixed continuous $h$,
\begin{align}\label{eq:area-rule}
 \E_{\mu_n}[F_n(h;\mathcal L)]=n^{(2A(h)-1/4+o(1))n^2}.
\end{align}
If $A(h)<1/8$, the probability that any basis satisfies all these prefix bounds is at most $n^{-(1/4-2A(h)+o(1))n^2}$.
\end{proposition}
\begin{proof}
For the complete rank sequence $(1,\ldots,n-1)$, Kim's formula~\cite[Theorem~5]{kimflags} has determinant powers all equal to two and gives
\begin{align}\label{eq:complete-chain-mean}
 \E_{\mu_n}[F_n(h;\mathcal L)]
 =\frac{n!}{2^{n-1}}\left(\prod_{d=2}^n\kappa_d\right)
 \exp\left(2n\ln n\sum_{k=1}^{n-1}h(k/n)\right).
\end{align}
Stirling's formula gives logarithm $-n^2\ln n/4+O(n^2)$ for the prefactor. A Riemann sum proves \eqref{eq:area-rule}, and Markov's inequality gives the probability bound.
\qed\end{proof}

Applying the proposition to $h=p+\varepsilon$ shows that any fixed continuous limiting profile occurring with probability bounded away from zero must have area at least $1/8$. The area also has a determinant-product interpretation. If the normalized prefix profiles converge uniformly to a continuous function $p$, then
\begin{align}
\int_0^1 p(x)\,dx = \lim_{n\to\infty} \frac{1}{n^2\ln n} \ln\!\left(\prod_{k=1}^{n-1}\det M_k\right).
\end{align}
Products of covolumes along nested chains also occur in the theory of flag heights and generalized Hermite constants~\cite{Meyer2009}.


\subsection{The universal lower bound}
\label{sec:any-basis}

The area rule alone gives a lower bound of $1/8$ on the maximum prefix determinant, after normalization by $n\ln n$. To obtain the stronger constant $(2-\sqrt2)/4$, we first replace the two outer parts of any hypothetical basis by reduced blocks. The reduced blocks admit curved upper bounds near the endpoints, while the middle prefix determinants remain unchanged.

\begin{proof}[Theorem~\ref{thm:any-basis}]
Fix $0<c<(2-\sqrt2)/4$, and suppose a basis satisfies $\max_kP_k\leq cn\ln n$. We keep its prefix chain between ranks $\lfloor2cn\rfloor$ and $\lfloor(1-2c)n\rfloor$, reduce the first block by reversed dual-HKZ, and reduce the final quotient by HKZ. The middle prefix determinants stay unchanged. Lemma~\ref{lem:hkz-bounds} bounds the new left flank by $x(1-\ln(x/(2c)))/2+o(1)$, and the right flank by its reflection. Thus the entire new profile lies below $g_c+o(1)$, where
\begin{align}\label{eq:plateau-family}
 u_c(x)=\min\{x,1-x,2c\},\qquad
 g_c(x)=\frac{u_c(x)}{2}\left(1-\ln\frac{u_c(x)}{2c}\right).
\end{align}
The plateau has height $c$ on $[2c,1-2c]$. Each flank has area $3(2c)^2/8$, so $A(g_c) = c(1-4c) + 3c^2 = c-c^2 < 1/8$. We can therefore choose a fixed $\varepsilon>0$ with $A(g_c+\varepsilon)<1/8$. For sufficiently large $n$, the modified basis satisfies this ceiling, which the area rule excludes with probability $1-o(1)$. Hence, with probability $1-o(1)$, no original basis has $\max_kP_k\leq cn\ln n$. Letting $c$ approach $(2-\sqrt2)/4$ sufficiently slowly as $n$ grows, we obtain, with probability $1-o(1)$:
\begin{align}\label{eq:universal-prefix-lower}
 \inf_{\mathbf B\text{ basis of }\mathcal L}\max_{0\leq k\leq n}P_k(\mathbf B)\geq\left(\frac{2-\sqrt2}{4}-o(1)\right)n\ln n.
\end{align}
The error is independent of the basis. Proposition~\ref{prop:tree-prefix} then proves Theorem~\ref{thm:any-basis}.
\qed\end{proof}


\subsection{Profile areas and the plateau}
\label{sec:plateau}

The HKZ/dual-HKZ profiles $p_{\mathrm H},p_{\mathrm D}$ both have area $1/8$. Integrating \eqref{eq:block-composition} gives
\begin{align}\label{eq:split-area}
 A(p)=\frac{q_\alpha}{2}+\alpha^2A(p_{\mathrm L})+(1-\alpha)^2A(p_{\mathrm R})=\frac{1}{8}
\end{align}
whenever both child areas are $1/8$. It follows that every upper profile obtained from a fixed finite binary construction, including the orange $128$-block example in Figure~\ref{fig:tree-profiles}, has this area. A first split at $1/2$ fixes $p(1/2)=1/8$, explaining the central dip retained by later refinements.

For comparison, the gray GSA profile gives:
\begin{align}\label{eq:gsa-area}
 p_{\mathrm{GSA}}(x)=\frac{x(1-x)}{2},\qquad
 A(p_{\mathrm{GSA}})=\frac{1}{12},\qquad
 \max_xp_{\mathrm{GSA}}(x)=\frac{1}{8}.
\end{align}
The area rule therefore excludes this profile on Haar-random lattices with probability $1-o(1)$. Its derivative $1/2-x$ is the gray line in Figure~\ref{fig:gs-profiles}.

Finally, consider the critical plateau profile, with $a=1-1/\sqrt2$. The red curve satisfies:
\begin{align}\label{eq:plateau-profile}
 p_{\mathrm{flat}}=g_{a/2},\qquad A(p_{\mathrm{flat}})=\frac{1}{8},\qquad
 \max_xp_{\mathrm{flat}}(x)=\frac{a}{2}=\frac{2-\sqrt2}{4}.
\end{align}
Its derivative, shown in Figure~\ref{fig:gs-profiles}, is
\begin{align}\label{eq:plateau-gs-profile}
 p_{\mathrm{flat}}'(x)=
 \begin{cases}
 \frac{1}{2}\ln(a/x),&0<x<a,\\
 0,&a\leq x\leq1-a,\\
 -\frac{1}{2}\ln(a/(1-x)),&1-a<x<1.
 \end{cases}
\end{align}
Under Assumption~\ref{heur:assumption}, the constructions in Section~\ref{sec:heuristic-upper} give upper bounds on the prefix profile that approach $p_{\mathrm{flat}}$, and attain its leading tree coefficient. This controls the prefix determinants, but does not imply that the individual Gram--Schmidt lengths converge to the derivative shown above. Whether the critical coefficient can be attained without a heuristic assumption remains open.


\section{Collision search with split bases}
\label{sec:fgl-plateau-note}

Using the search radius $R = \sqrt n$ in the Feng--Gao--Liu (FGL) sampling and collision procedure~\cite{fgl2026} gives the geometric product
\begin{equation}
  \mathcal P(\mathbf B)=\prod_{i=1}^n\max\{1,a_i^{-1}\}.
  \label{eq:fgl-product}
\end{equation}
The dominant geometric factor in the FGL running-time bound is $\sqrt{\mathcal P(\mathbf B)}$. We first bound this factor from below for all bases, and then determine it for the two-block construction. To obtain an algorithmic upper bound, we also control the cost of sampling from the chosen basis. Theorem~\ref{thm:fgl-lower-note} concerns the geometric factor in the analysis; it does not give a lower bound on the actual running time of the procedure.

The key difference with enumeration is that collision search depends on the total variation of the prefix profile, whereas enumeration depends on its maximum. To make this precise, write $[u]_+=\max\{u,0\}$, let $\operatorname{TV}$ denote total variation on $[0,1]$, and let $p_n$ linearly interpolate $p_n(k/n)=P_k/(n\ln n)$.


\subsection{A lower bound for all bases}
\label{sec:fgl-lower}

\begin{proof}[Theorem~\ref{thm:fgl-lower-note}]
Since $\sum_i y_i=0$,
\begin{equation}
  \frac{\ln\sqrt{\mathcal P(\mathbf B)}}{n\ln n}
  =\frac{\sum_i[-y_i]_+}{2n\ln n}
  =\frac14\operatorname{TV}(p_n)
  \ge\frac{\max_kP_k}{2n\ln n}.
  \label{eq:fgl-variation}
\end{equation}
The universal prefix lower bound \eqref{eq:universal-prefix-lower} gives $\max_kP_k\ge((2-\sqrt2)/4-o(1))n\ln n$, simultaneously for all bases, with probability $1-o(1)$. Because both the error term and the inequality above are uniform over the bases, the resulting lower bound holds simultaneously for all bases, with the leading constant divided by two. Multiplying the search radius by any fixed positive constant changes the logarithm of the cost factor by at most $O(n)$, uniformly in the basis. The same leading lower bound therefore applies to the radius used for CVPP below.
\qed\end{proof}


\subsection{The two-block cost factor}
\label{sec:fgl-two-block}

\begin{proof}[Theorem~\ref{thm:fgl-two-way}]
The normalized HKZ blocks used in Section~\ref{sec:two-way} have dimensions $d=n/2+O(1)$ and shortest-vector lengths $\rho_1\asymp\sqrt d$, by the shortest-vector properties imposed in the construction. Let $\rho_i$ denote the Gram--Schmidt lengths of either normalized block. The HKZ product bound~\cite[Lemma~1]{fgl2026} gives, for every subset $I$ of size $k$,
\begin{align}
  -\sum_{i\in I}\ln\rho_i
  \le \frac{k}{2}\ln(d/k)\ln n+O(n).
\end{align}
Here and below the error is uniform in $I$, and $0\ln(d/0)=0$. Consequently,
\begin{equation}
 \begin{split}
  \sum_{i=1}^d\left[\frac14\ln n-\ln\rho_i\right]_+
  &\le \ln n\max_{0\le k\le d}
       \left\{\frac{k}{4}+\frac{k}{2}\ln(d/k)\right\}+O(n) =\left(\frac1{2\sqrt e}+o(1)\right)d\ln n.
 \end{split}
 \label{eq:fgl-hkz-subset}
\end{equation}
Rescaling the second block multiplies its lengths by $n^{-1/4}e^{O(1)}$, so the estimate above bounds its contribution to $\sum_i[-y_i]_+$. The first block is reversed and dualized, then scaled by $n^{1/4}e^{O(1)}$. Since $\sum_i\ln\rho_i=0$, its contribution is bounded by $(1/(2\sqrt e)-1/4+o(1))d\ln n$. Adding both contributions and dividing by $2n\ln n$ proves the upper bound in~\eqref{eq:fgl-two-way}. In fact, these estimates give $\ln\mathcal P(\mathbf B_2)\le2\tau_2n\ln n+O(n)$.

For the reverse inequality, Section~\ref{sec:two-way} gives uniform convergence $p_n\to p_2$. The two peaks of $p_2$ have height $1/(4\sqrt e)$, the intervening dip has height $1/8$, and its endpoints are zero. Evaluating the variation along these five points gives $\operatorname{TV}(p_n)\ge1/\sqrt e-1/4-o(1)$. Equation~\eqref{eq:fgl-variation} therefore gives a matching lower bound, showing that the constructed basis satisfies $\sqrt{\mathcal P(\mathbf B_2)}=n^{\tau_2 n+o(n)}$.
\qed\end{proof}

\paragraph{Search on the preprocessed basis.}
To turn the product estimate into a search bound, we must control the sampling cost and choose an appropriate radius. Applying the same HKZ product bound with $k=1$ and $k=d-1$, and using $\prod_i\rho_i=1$, gives $\max_i\max\{1,a_i^{-1}\}\le e^{O(n)}=n^{o(n)}$ for $\mathbf B_2$. This makes the cost of each sample $n^{o(n)}$.

For the radius, we cannot use the length of the first basis vector: by \eqref{eq:split-first-vector}, this length is not within a constant factor of $\lambda_1(\mathcal L)$. We therefore replace the cutoff $\|\mathbf b_1\|$ in \cite[Sections~3--4]{fgl2026} by $R=\sqrt n$ and retain the shortest nonzero difference found. Minkowski gives $\lambda_1(\mathcal L)\le R$, while the Haar shortest-vector bound gives $R/\lambda_1(\mathcal L)=O(1)$ with high probability. The short-translation argument therefore still applies, and grid cells contain at most $\exp(O(n\ln\ln n))$ lattice points by the same packing bound. Together with the sampling bound above, the collision analysis gives $n^{o(n)}\sqrt{\mathcal P(\mathbf B_2)}=n^{\tau_2 n+o(n)}$ time and polynomial space, with constant success probability and polynomial encoding factors suppressed. This bounds the search time once the basis is available; constructing the basis is part of preprocessing.


\subsection{Affine collision search for CVPP}
\label{sec:fgl-cvpp}

\begin{proof}[Theorem~\ref{thm:cvpp-collision}]
We use the two-block basis from Section~\ref{sec:two-way}, together with the product and coordinate estimates proved above. With probability $1-o(1)$, this basis exists and the shortest-vector bounds in \eqref{eq:short-tail} hold for both the lattice and its dual. We work on this event, on which $\lambda_1(\mathcal L)\asymp\sqrt n$. By Section~\ref{sec:prefix-nodes}, the fixed radius $R=10\sqrt n$ contains a closest vector to every target. To account for this radius, define:
\begin{equation}
 \mathcal P_R(\mathbf B)=\prod_i\max\left\{1,\frac{R}{\sqrt n\,a_i}\right\}.
 \label{eq:fgl-radius-product}
\end{equation}
Thus $\mathcal P(\mathbf B)\le\mathcal P_R(\mathbf B)\le 10^n\mathcal P(\mathbf B)$, and the maximum coordinate scale remains $n^{o(n)}$.

Fix any target $\boldsymbol t$. If $\boldsymbol t\in\mathcal L$, return $\boldsymbol t$. Otherwise, let $\boldsymbol z^*$ be a closest lattice point and put $\boldsymbol v^*=\boldsymbol t-\boldsymbol z^*$, so $\|\boldsymbol v^*\|\le R$. We sample one implicit array in $\mathcal L$ and another in $\boldsymbol t+\mathcal L$, both using rounded Gram--Schmidt coordinates centered at the origin:
\begin{align}
 \kappa_i(\boldsymbol x)=
 \left\lfloor\frac{\langle \boldsymbol x,\boldsymbol b_i^*\rangle}{a_i^2}\right\rceil.
\end{align}
Here $\lfloor u\rceil=\lfloor u+1/2\rfloor$. Restricted to either coset, $\kappa$ is a bijection with $\mathbb Z^n$. Indeed, for $\boldsymbol x=\boldsymbol t+\mathbf B\boldsymbol z$, each integer $z_i$ is recovered uniquely from $\kappa_i(\boldsymbol x)$ once $z_{i+1},\ldots,z_n$ are known. The same applies with $\boldsymbol t=\boldsymbol0$. We apply the FGL product distribution to these integer coordinates, using its rounded weights $\alpha_i\le\min\{1,na_i^2/R^2\}$. This gives distributions $\nu_0$ on $\mathcal L$ and $\nu_{\boldsymbol t}$ on $\boldsymbol t+\mathcal L$.

Translation by $\boldsymbol v^*$ maps $\mathcal L$ bijectively onto $\boldsymbol t+\mathcal L$. The rounding error in the difference of each coordinate is at most one, and hence, for every $\boldsymbol x\in \mathcal L$,
\begin{align}
 \sum_i\alpha_i
       |\kappa_i(\boldsymbol x+\boldsymbol v^*)-\kappa_i(\boldsymbol x)|^2
 \le \frac{2n}{R^2}\|\boldsymbol v^*\|^2+2n\le4n.
\end{align}
The central-box argument of~\cite[Theorem~3.1]{fgl2026} therefore applies to these two coset bijections without change. It gives maximum point masses at most $2^{O(n)}/\mathcal P_R(\mathbf B)$ and
\begin{equation}
 \sum_{\boldsymbol x\in \mathcal L}\nu_0(\boldsymbol x)\nu_{\boldsymbol t}(\boldsymbol x+\boldsymbol v^*)
 \ge\frac{2^{-O(n)}}{\mathcal P_R(\mathbf B)}.
 \label{eq:fgl-affine-overlap}
\end{equation}
Neither sampler requires knowledge of $\boldsymbol z^*$ or $\boldsymbol v^*$. Independent arrays of length $2^{O(n)}\mathcal P_R(\mathbf B)$, generated from short seeds as in FGL, consequently have the same representation and second-moment bounds, up to exponential factors. Those arguments use only the point-mass and overlap bounds and remain valid for the two different marginal distributions.

For rationally encoded bases and targets, a common scaling clears their denominators and places both cosets in $\mathbb Z^n$. We scale the search radius by the same factor, leaving the coordinate scales and all ratios used below unchanged. This permits the exact discrete-grid implementation of FGL for arbitrary encoded targets.

We then apply the FGL random grid at scale $W=O(R\sqrt{\ln n})$. A pair with difference $\boldsymbol v^*$ shares a cell with probability $2^{-O(n)}$. Packing each coset separately bounds the number of distinct points of their union in a cell by
\begin{align}
 2\left(1+\frac{2W}{\lambda_1(\mathcal L)}\right)^n
 =\exp(O(n\ln\ln n)).
\end{align}
This bound is uniform in $\boldsymbol t$, even when the two cosets are very close. The collision lemma~\cite[Lemma~4]{fgl2026} can therefore be applied with the good pairs restricted to cross-array pairs whose difference is $\boldsymbol v^*$. We initialize the retained candidate with $\boldsymbol0$ and ignore pairs from the same array. For every reported cross-array pair $X\in\mathcal L$, $Y\in\boldsymbol t+\mathcal L$, the vector $\boldsymbol t-(Y-X)$ belongs to $\mathcal L$. Throughout the prescribed search budget, we retain the candidate closest to $\boldsymbol t$. With constant probability a pair with difference $\boldsymbol v^*$ is reported, making the retained candidate exact. A constant number of independent runs raises this probability above $2/3$.

As in FGL, the resulting query time is at most $n^{o(n)}\sqrt{\mathcal P_R(\mathbf B)} =n^{(\tau_2+o(1))n}$, with polynomial space. Every geometric bound is independent of the target, so one basis supports all targets, with fresh randomness for each query.

It remains to describe the preprocessing and the size of the retained advice. Preprocessing may enumerate primitive half-rank sublattices, perform the two block reductions, and test the resulting product and coordinate bounds. The estimates above have $e^{O(n)}$ slack: for example, thresholds $\mathcal P(\mathbf B)\le n^{2\tau_2n}e^{n\sqrt{\ln n}}$ and $\max_i\max\{1,a_i^{-1}\}\le e^{n\sqrt{\ln n}}$ therefore suffice for all sufficiently large $n$. Rounding these thresholds upward to powers of two makes the tests exact and changes the query exponent only by $o(1)$.

The product bound also gives polynomial encoding length after size reduction: since $\sum_i y_i=0$, one has $|y_i|\leq\ln\mathcal P(\mathbf B)=O(n\ln n)$ and $\|\boldsymbol b_i\|\leq\sqrt n\,\mathcal P(\mathbf B)$. As in Section~\ref{sec:prefix-nodes}, the corresponding integer basis changes have polynomial bit length relative to the input. One may therefore restrict preprocessing to a finite search over such changes; the bounds above guarantee success on the stated event. Only the chosen basis is retained as advice. Preprocessing has unrestricted running time and is independent of the target.
\qed\end{proof}


\subsection{Other split profiles}
\label{sec:fgl-other-profiles}

\paragraph{Further binary splits.}
The recursive splits in Section~\ref{sec:recursive-splits} were chosen to minimize the largest prefix determinant. For collision search, however, we need to minimize $\operatorname{TV}(p)/4$. Evaluating the binary profiles of Section~\ref{sec:recursive-splits} gives collision coefficients approximately $0.09504173$ for four blocks and $0.11758925$ for $16\,384$ blocks, compared with $0.08913266$ for two blocks. Continuing the balanced refinement suggests that the collision coefficient approaches $1/8$, even though the enumeration coefficient decreases towards $0.147397\ldots$. The limits refer to the analytic profiles, taking $n\to\infty$ at each fixed depth before increasing the depth. The reason is that further splits lower the largest prefix peak while introducing additional rises and falls. These particular splits therefore do not appear to improve collision search.

\paragraph{The heuristic block construction.}
The analytic profiles of Section~\ref{sec:heuristic-upper} also give a larger collision coefficient than the two-block construction. The limiting coefficient is $(2-\sqrt2)/8+(\sqrt2-1)/(4e)=0.11131847\ldots$. In fact, in the ideal block-profile model it already has this value for every fixed total number $K\geq3$ of blocks, with $K-2$ middle blocks. Smaller blocks reduce the height of each prefix ripple, but their total variation stays unchanged. This value exceeds both the single-split coefficient and the HKZ coefficient $1/(4e)$, although it is below the predicted limit $1/8$ of the deeply refined binary profiles.

These are calculations for the analytic profiles: the first-moment assumption in Section~\ref{sec:heuristic-upper} supplies prefix ceilings and does not by itself establish the corresponding Gram--Schmidt variation or search cost. These comparisons concern the specific enumeration constructions above; they do not establish optimality of the two-block basis among all configurations allowed by Assumption~\ref{heur:assumption}. Constructing bases specifically optimized for collision search, and determining whether its lower bound $(2-\sqrt2)/8$ is attainable, remain questions for future work.


\bibliographystyle{splncs04}
\bibliography{references}

@inproceedings{ajtai01,
  series = {STOC01},
  title = {{A sieve algorithm for the shortest lattice vector problem}},
  doi = {10.1145/380752.380857},
  booktitle = {Proceedings of the thirty-third annual ACM symposium on Theory of computing},
  publisher = {ACM},
  author = {Ajtai, Mikl{\'{o}}s and Kumar, Ravi and Sivakumar, D.},
  year = {2001},
  pages = {601--610}
}

@inproceedings{bdgl,
  title = {{New directions in nearest neighbor searching with applications to lattice sieving}},
  doi = {10.1137/1.9781611974331.ch2},
  booktitle = {Proceedings of the Twenty-Seventh Annual ACM-SIAM Symposium on Discrete Algorithms},
  publisher = {Society for Industrial and Applied Mathematics},
  author = {Becker, Anja and Ducas, L{\'{e}}o and Gama, Nicolas and Laarhoven, Thijs},
  year = {2016},
  pages = {10--24}
}

@inproceedings{chennguyen,
  title = {{BKZ 2.0: Better Lattice Security Estimates}},
  doi = {10.1007/978-3-642-25385-0_1},
  booktitle = {Advances in Cryptology -- ASIACRYPT 2011},
  publisher = {Springer Berlin Heidelberg},
  author = {Chen, Yuanmi and Nguyen, Phong Q.},
  year = {2011},
  pages = {1--20}
}

@inproceedings{cvpp-bloemer00,
  title = {{Closest Vectors, Successive Minima, and Dual HKZ-Bases of Lattices}},
  doi = {10.1007/3-540-45022-x_22},
  booktitle = {Automata, Languages and Programming},
  publisher = {Springer Berlin Heidelberg},
  author = {Bl{\"{o}}mer, Johannes},
  year = {2000},
  pages = {248--259}
}

@article{cvpp-bls16,
  title = {{Tuple lattice sieving}},
  volume = {19},
  doi = {10.1112/s1461157016000292},
  number = {A},
  journal = {LMS Journal of Computation and Mathematics},
  publisher = {Wiley},
  author = {Bai, Shi and Laarhoven, Thijs and Stehl{\'{e}}, Damien},
  year = {2016},
  pages = {146--162}
}

@article{cvpp-fu16,
  title = {{Improvements in the analysis of Kannan's CVP algorithm}},
  volume = {621},
  doi = {10.1016/j.tcs.2016.01.038},
  journal = {Theoretical Computer Science},
  publisher = {Elsevier BV},
  author = {Fu, Lianzhi},
  year = {2016},
  pages = {103--112}
}

@article{cvpp-laarhoven21,
  title = {{Approximate Voronoi cells for lattices, revisited}},
  volume = {15},
  doi = {10.1515/jmc-2020-0074},
  number = {1},
  journal = {Journal of Mathematical Cryptology},
  publisher = {Walter de Gruyter GmbH},
  author = {Laarhoven, Thijs},
  year = {2021},
  pages = {60--71}
}

@inproceedings{dadushmicciancio,
  title = {{Algorithms for the Densest Sub-Lattice Problem}},
  doi = {10.1137/1.9781611973105.79},
  booktitle = {Proceedings of the Twenty-Fourth Annual ACM-SIAM Symposium on Discrete Algorithms},
  publisher = {Society for Industrial and Applied Mathematics},
  author = {Dadush, Daniel and Micciancio, Daniele},
  year = {2013},
  pages = {1103--1122}
}

@article{ducasloyer,
  title = {{Lattice Reduction via Dense Sublattices: A Cryptanalytic No-Go}},
  volume = {2},
  doi = {10.62056/ae0fhsfg},
  number = {3},
  journal = {IACR Communications in Cryptology},
  publisher = {International Association for Cryptologic Research},
  author = {Ducas, L{\'e}o and Loyer, Johanna},
  year = {2025},
  pages = {12:1--12:15}
}

@article{finckepohst,
  title = {{Improved methods for calculating vectors of short length in a lattice, including a complexity analysis}},
  volume = {44},
  doi = {10.1090/s0025-5718-1985-0777278-8},
  number = {170},
  journal = {Mathematics of Computation},
  publisher = {American Mathematical Society (AMS)},
  author = {Fincke, Ulrich and Pohst, Michael},
  year = {1985},
  pages = {463--471}
}

@inproceedings{g6k,
  title = {{The General Sieve Kernel and New Records in Lattice Reduction}},
  doi = {10.1007/978-3-030-17656-3_25},
  booktitle = {Advances in Cryptology -- EUROCRYPT 2019},
  publisher = {Springer International Publishing},
  author = {Albrecht, Martin R. and Ducas, L{\'{e}}o and Herold, Gottfried and Kirshanova, Elena and Postlethwaite, Eamonn W. and Stevens, Marc},
  year = {2019},
  pages = {717--746}
}

@inproceedings{gnr,
  title = {{Lattice Enumeration Using Extreme Pruning}},
  doi = {10.1007/978-3-642-13190-5_13},
  booktitle = {Advances in Cryptology -- EUROCRYPT 2010},
  publisher = {Springer Berlin Heidelberg},
  author = {Gama, Nicolas and Nguyen, Phong Q. and Regev, Oded},
  year = {2010},
  pages = {257--278}
}

@article{haque,
  title = {{Analysing recursive preprocessing of BKZ lattice reduction algorithm}},
  volume = {11},
  doi = {10.1049/iet-ifs.2016.0049},
  number = {2},
  journal = {IET Information Security},
  publisher = {Institution of Engineering and Technology (IET)},
  author = {Haque, Md. Mokammel and Pieprzyk, Josef},
  year = {2017},
  pages = {114--120}
}

@inproceedings{hs2007,
  title = {{Improved Analysis of Kannan's Shortest Lattice Vector Algorithm}},
  doi = {10.1007/978-3-540-74143-5_10},
  booktitle = {Advances in Cryptology - CRYPTO 2007},
  publisher = {Springer Berlin Heidelberg},
  author = {Hanrot, Guillaume and Stehl{\'{e}}, Damien},
  pages = {170--186},
  year = {2007}
}

@inproceedings{kannan83,
  series = {STOC '83},
  title = {{Improved algorithms for integer programming and related lattice problems}},
  doi = {10.1145/800061.808749},
  booktitle = {Proceedings of the fifteenth annual ACM symposium on Theory of computing  - STOC '83},
  publisher = {ACM Press},
  author = {Kannan, Ravi},
  year = {1983},
  pages = {193--206}
}

@article{kannan87,
  title = {{Minkowski's Convex Body Theorem and Integer Programming}},
  volume = {12},
  doi = {10.1287/moor.12.3.415},
  number = {3},
  journal = {Mathematics of Operations Research},
  publisher = {Institute for Operations Research and the Management Sciences (INFORMS)},
  author = {Kannan, Ravi},
  year = {1987},
  pages = {415--440}
}

@article{kimflags,
  title = {{Mean value formulas on sublattices and flags of the random lattice}},
  volume = {241},
  doi = {10.1016/j.jnt.2022.03.013},
  journal = {Journal of Number Theory},
  publisher = {Elsevier BV},
  author = {Kim, Seungki},
  year = {2022},
  pages = {330--351}
}

@article{linguyen14,
  title = {{Approximating the densest sublattice from Rankin's inequality}},
  volume = {17},
  doi = {10.1112/s1461157014000333},
  number = {A},
  journal = {LMS Journal of Computation and Mathematics},
  publisher = {Wiley},
  author = {Li, Jianwei and Nguyen, Phong Q.},
  year = {2014},
  pages = {92--111}
}

@article{lll,
  title = {{Factoring polynomials with rational coefficients}},
  volume = {261},
  doi = {10.1007/bf01457454},
  number = {4},
  journal = {Mathematische Annalen},
  publisher = {Springer Science and Business Media LLC},
  author = {Lenstra, A. K. and Lenstra, H. W. and Lov{\'{a}}sz, L.},
  year = {1982},
  pages = {515--534}
}

@incollection{nguyenhermite,
  title = {{Hermite's Constant and Lattice Algorithms}},
  doi = {10.1007/978-3-642-02295-1_2},
  booktitle = {The LLL Algorithm},
  publisher = {Springer Berlin Heidelberg},
  author = {Nguyen, Phong Q.},
  year = {2009},
  pages = {19--69}
}

@article{schnorreuchner,
  title = {{Lattice basis reduction: Improved practical algorithms and solving subset sum problems}},
  volume = {66},
  doi = {10.1007/bf01581144},
  number = {1-3},
  journal = {Mathematical Programming},
  publisher = {Springer Science and Business Media LLC},
  author = {Schnorr, C. P. and Euchner, M.},
  year = {1994},
  pages = {181--199}
}

@inproceedings{schnorrgsa,
  title = {{Lattice Reduction by Random Sampling and Birthday Methods}},
  doi = {10.1007/3-540-36494-3_14},
  booktitle = {STACS 2003},
  publisher = {Springer Berlin Heidelberg},
  author = {Schnorr, Claus Peter},
  year = {2003},
  pages = {145--156}
}

@article{siegel,
  title = {{A Mean Value Theorem in Geometry of Numbers}},
  volume = {46},
  doi = {10.2307/1969027},
  number = {2},
  journal = {The Annals of Mathematics},
  publisher = {JSTOR},
  author = {Siegel, Carl Ludwig},
  year = {1945},
  pages = {340--347}
}

@article{sodergren,
  title = {{On the Poisson distribution of lengths of lattice vectors in a random lattice}},
  volume = {269},
  doi = {10.1007/s00209-010-0772-8},
  number = {3-4},
  journal = {Mathematische Zeitschrift},
  publisher = {Springer Science and Business Media LLC},
  author = {S{\"{o}}dergren, Anders},
  year = {2011},
  pages = {945--954}
}

@article{strom,
  title = {{On the probability of a random lattice avoiding a large convex set}},
  volume = {103},
  doi = {10.1112/plms/pdr021},
  number = {6},
  journal = {Proceedings of the London Mathematical Society},
  publisher = {Wiley},
  author = {Str{\"{o}}mbergsson, Andreas},
  year = {2011},
  pages = {950--1006}
}

@misc{fgl2026,
  author = {Yansong Feng and Yiming Gao and Jiaqi Liu},
  title = {{Faster {SVP} in Polynomial Space}},
  year = {2026},
  note = {arXiv:2609.21612},
  url = {https://arxiv.org/abs/2609.21612}
}

@misc{hscomplete,
  author = {Guillaume Hanrot and Damien Stehl{\'e}},
  title = {{A Complete Worst-Case Analysis of {Kannan}'s Shortest Lattice Vector Algorithm}},
  year = {2010},
  note = {Extended manuscript},
  url = {https://perso.ens-lyon.fr/damien.stehle/downloads/KANNAN_EXTENDED.pdf}
}

@techreport{hsworst,
  author = {Guillaume Hanrot and Damien Stehl{\'e}},
  title = {{Worst-Case {Hermite--Korkine--Zolotarev} Reduced Lattice Bases}},
  institution = {INRIA},
  number = {RR-6422},
  year = {2008},
  url = {https://arxiv.org/abs/0801.3331}
}

@inproceedings{cvpp-drs14,
  author = {Daniel Dadush and Oded Regev and Noah Stephens-Davidowitz},
  title = {{On the Closest Vector Problem with a Distance Guarantee}},
  booktitle = {Proceedings of the 29th Conference on Computational Complexity},
  pages = {98--109},
  year = {2014},
  url = {https://arxiv.org/abs/1409.8063}
}

@misc{svpchallenge,
  author = {Michael Schneider and Nicolas Gama},
  title = {{{SVP Challenge}}},
  year = {2010},
  note = {Benchmark website; accessed 30 September 2026},
  url = {https://www.latticechallenge.org/svp-challenge/}
}

@book{folland,
  author = {Gerald B. Folland},
  title = {{Real Analysis: Modern Techniques and Their Applications}},
  edition = {2},
  publisher = {John Wiley \& Sons},
  year = {1999}
}

@inproceedings{faster,
  author = {Martin R. Albrecht and Shi Bai and Pierre-Alain Fouque and Paul Kirchner and Damien Stehl{\'e} and Weiqiang Wen},
  title = {{Faster Enumeration-Based Lattice Reduction: Root {Hermite} Factor $k^{1/(2k)}$ in Time $k^{k/8+o(k)}$}},
  booktitle = {Advances in Cryptology -- CRYPTO 2020, Part II},
  volume = {12171},
  series = {Lecture Notes in Computer Science},
  pages = {186--212},
  publisher = {Springer},
  year = {2020},
  doi = {10.1007/978-3-030-56880-1_7}
}

@inproceedings{anss,
  author = {Yoshinori Aono and Phong Q. Nguyen and Takenobu Seito and Junji Shikata},
  title = {{Lower Bounds on Lattice Enumeration with Extreme Pruning}},
  booktitle = {Advances in Cryptology -- CRYPTO 2018, Part II},
  volume = {10992},
  series = {Lecture Notes in Computer Science},
  pages = {608--637},
  publisher = {Springer},
  year = {2018},
  doi = {10.1007/978-3-319-96881-0_21}
}

@article{banaszczyk,
  author = {Wojciech Banaszczyk},
  title = {{New Bounds in Some Transference Theorems in the Geometry of Numbers}},
  journal = {Mathematische Annalen},
  volume = {296},
  pages = {625--635},
  year = {1993},
  doi = {10.1007/BF01445125}
}

@inproceedings{DLW19,
  author = {Emmanouil Doulgerakis and Thijs Laarhoven and Benne de Weger},
  title = {{Finding Closest Lattice Vectors Using Approximate {Voronoi} Cells}},
  editor = {Jintai Ding and Rainer Steinwandt},
  booktitle = {Post-Quantum Cryptography -- PQCrypto 2019},
  series = {Lecture Notes in Computer Science},
  volume = {11505},
  pages = {3--22},
  publisher = {Springer},
  year = {2019},
  doi = {10.1007/978-3-030-25510-7_1}
}

@article{helfrich,
  author = {Bettina Helfrich},
  title = {{Algorithms to Construct {Minkowski} Reduced and {Hermite} Reduced Lattice Bases}},
  journal = {Theoretical Computer Science},
  volume = {41},
  number = {2--3},
  pages = {125--139},
  year = {1985},
  doi = {10.1016/0304-3975(85)90067-2}
}

@article{cvpp-hrs20,
  author = {Christoph Hunkenschr{\"o}der and Gina Reuland and Matthias Schymura},
  title = {{On Compact Representations of {Voronoi} Cells of Lattices}},
  journal = {Mathematical Programming},
  volume = {183},
  pages = {337--358},
  year = {2020},
  doi = {10.1007/s10107-019-01463-3},
  url = {https://arxiv.org/abs/1811.08532}
}

@article{kz1873,
  author = {A. Korkine and G. Zolotareff},
  title = {{Sur les formes quadratiques}},
  journal = {Mathematische Annalen},
  volume = {6},
  number = {3},
  pages = {366--389},
  year = {1873},
  doi = {10.1007/BF01442795}
}

@article{cvpp-wagnermaze14,
  author = {Urs Wagner and G{\'e}rard Maze},
  title = {{Improvements in Closest Point Search Based on Dual {HKZ}-Bases}},
  journal = {Theoretical Computer Science},
  volume = {520},
  pages = {62--74},
  year = {2014},
  doi = {10.1016/j.tcs.2013.10.005},
  url = {https://arxiv.org/abs/1201.5273}
}

@inproceedings{hps,
  author = {Guillaume Hanrot and Xavier Pujol and Damien Stehl{\'e}},
  title = {Analyzing Blockwise Lattice Algorithms Using Dynamical Systems},
  booktitle = {Advances in Cryptology -- {CRYPTO} 2011},
  editor = {Phillip Rogaway},
  series = {Lecture Notes in Computer Science},
  volume = {6841},
  pages = {447--464},
  publisher = {Springer},
  year = {2011},
  doi = {10.1007/978-3-642-22792-9_25}
}

@inproceedings{aeps2025,
  author    = {Divesh Aggarwal and Thomas Espitau and
               Spencer Peters and Noah Stephens-Davidowitz},
  title     = {Recursive Lattice Reduction---A Framework for
               Finding Short Lattice Vectors},
  booktitle = {Proceedings of the 2025 Symposium on Simplicity
               in Algorithms ({SOSA})},
  pages     = {533--573},
  publisher = {SIAM},
  year      = {2025},
  doi       = {10.1137/1.9781611978315.40},
  url       = {https://arxiv.org/abs/2311.15064}
}

@inproceedings{ghkn2006,
  author    = {Nicolas Gama and Nick Howgrave-Graham and
               Henrik Koy and Phong Q. Nguyen},
  title     = {{Rankin}'s Constant and Blockwise Lattice Reduction},
  booktitle = {Advances in Cryptology -- {CRYPTO} 2006},
  series    = {Lecture Notes in Computer Science},
  volume    = {4117},
  pages     = {112--130},
  publisher = {Springer},
  year      = {2006},
  doi       = {10.1007/11818175_7}
}

@inproceedings{yuducas2018,
  author    = {Yang Yu and L{\'e}o Ducas},
  title     = {Second Order Statistical Behavior of {LLL} and {BKZ}},
  booktitle = {Selected Areas in Cryptography -- {SAC} 2017},
  series    = {Lecture Notes in Computer Science},
  volume    = {10719},
  pages     = {3--22},
  publisher = {Springer},
  year      = {2018},
  doi       = {10.1007/978-3-319-72565-9_1}
}

@article{shapiraweiss2014,
  author  = {Uri Shapira and Barak Weiss},
  title   = {A Volume Estimate for the Set of Stable Lattices},
  journal = {Comptes Rendus Math{\'e}matique},
  volume  = {352},
  number  = {11},
  pages   = {875--879},
  year    = {2014},
  doi     = {10.1016/j.crma.2014.08.019}
}

@article{LiNguyen2025,
  author  = {Li, Jianwei and Nguyen, Phong Q.},
  title   = {A Complete Analysis of the {BKZ} Lattice Reduction Algorithm},
  journal = {Journal of Cryptology},
  volume  = {38},
  year    = {2025},
  note    = {Article 12},
  doi     = {10.1007/s00145-024-09527-0}
}

@article{Meyer2009,
  author  = {Meyer, Bertrand},
  title   = {Generalised {Hermite} constants, {Voronoi} theory
             and heights on flag varieties},
  journal = {Bulletin de la Soci{\'e}t{\'e} Math{\'e}matique de France},
  volume  = {137},
  number  = {1},
  pages   = {127--158},
  year    = {2009},
  doi     = {10.24033/bsmf.2571}
}

@article{Thunder1998,
  author  = {Thunder, Jeffrey Lin},
  title   = {Higher-dimensional analogs of {Hermite}'s constant},
  journal = {Michigan Mathematical Journal},
  volume  = {45},
  number  = {2},
  pages   = {301--314},
  year    = {1998},
  doi     = {10.1307/mmj/1030132184}
}

@article{Neumaier2017,
  author  = {Neumaier, Arnold},
  title   = {Bounding basis reduction properties},
  journal = {Designs, Codes and Cryptography},
  volume  = {84},
  pages   = {237--259},
  year    = {2017},
  doi     = {10.1007/s10623-016-0273-9}
}

\end{document}